\documentclass[runningheads, envcountsame]{llncs}
 
\usepackage{geometry}%
\usepackage{appendix}
\usepackage{amsmath}
\usepackage{amsfonts}
\usepackage{amssymb}
\usepackage{ifpdf}
\usepackage{algorithm}
\usepackage{algpseudocode}
\usepackage{comment}
\usepackage{subfig}
\usepackage{wrapfig}
\usepackage{cite}
\usepackage{colortbl}
\usepackage{tikz}
\usepackage{morefloats}

\usepackage{caption}
\authorrunning{J.D.L., J.H., M.J.P., I.P., and M.S.}

\spnewtheorem{claim}{Claim}{\itshape}{\normalfont}

\usepackage{commands-tam}

\ifpdf

  \usepackage[pdftex]{epsfig}
  \usepackage[pdftex]{hyperref}

\else

    \usepackage[dvips]{epsfig}
    \newcommand{\href}[2]{#2}

\fi

\newif\ifabstract
\newif\iffull

\ifabstract
	\fullfalse
\else
	\fulltrue
\fi

\title{Self-Assembly of 3-D Structures Using 2-D Folding Tiles}

\author{
 J{\'{e}}r{\^{o}}me Durand{-}Lose
   \thanks{LIX, Ecole Polytechnique, UMR 7161, F-91128 Palaiseau Cedex, France and LIFO, Universit\'e d'Orl\'eans,  \'EA 4022, F-45067 Orl\'eans, France
   \protect\url{jerome.durand-lose@univ-orleans.fr}}
\and
 Jacob Hendricks%
    \thanks{Department of Computer Science and Information Systems, University of Wisconsin - River Falls, River Falls, WI, USA
    \protect\url{jacob.hendricks@uwrf.edu}}
\and
 Matthew J. Patitz
    \thanks{Department of Computer Science and Computer Engineering, University of Arkansas, Fayetteville, AR, USA
    \protect\url{patitz@uark.edu}  This author's research was supported in part by National Science Foundation Grant CCF-1422152 and CAREER-1553166.}
\and
 Ian Perkins
    \thanks{Department of Computer Science and Computer Engineering, University of Arkansas, Fayetteville, AR, USA
    \protect\url{irperkin@uark.edu}.}
\and
Michael Sharp
    \thanks{Department of Computer Science and Computer Engineering, University of Arkansas, Fayetteville, AR, USA.
    \protect\url{mrs018@uark.edu}.
    This author's research was supported in part by National Science Foundation Grants CCF-1422152 and CAREER-1553166.}
}

\institute{}
\date{}

\begin{document}

\maketitle

\vspace{-20pt}
\begin{abstract}
Self-assembly is a process which is ubiquitous in natural, especially biological systems.  It occurs when groups of relatively simple components spontaneously combine to form more complex structures.  While such systems have inspired a large amount of research into designing theoretical models of self-assembling systems, and even laboratory-based implementations of them, these artificial models and systems often tend to be lacking in one of the powerful features of natural systems (e.g. the assembly and folding of proteins), namely the dynamic reconfigurability of structures.  In this paper, we present a new mathematical model of self-assembly, based on the abstract Tile Assembly Model (aTAM), called the Flexible Tile Assembly Model (FTAM). In the FTAM, the individual components are 2-dimensional square tiles as in the aTAM, but in the FTAM, bonds between the edges of tiles can be flexible, allowing bonds to flex and entire structures to reconfigure, thus allowing 2-dimensional components to form 3-dimensional structures.  We analyze the powers and limitations of FTAM systems by (1) demonstrating how flexibility can be controlled to carefully build desired structures, and (2) showing how flexibility can be beneficially harnessed to form structures which can ``efficiently'' reconfigure into many different configurations and/or greatly varying configurations.  We also show that with such power comes a heavy burden in terms of computational complexity of simulation and prediction by proving that, for important properties of FTAM systems, determining their existence is intractable, even for properties which are easily computed for systems in less dynamic models.

\end{abstract}

\section{Introduction}

Proteins are a fantastically diverse set of biomolecules, with structures and functions that can vary wildly from each other, such as fibrous proteins (like collagen), enzymatic proteins (like catalase), and transport proteins (like hemoglobin). Truly amazing is the fact that such diversity arises solely from the linear combination of only $20$ amino acid building blocks.  It is the specific sequence of amino acids, interacting with each other as they are combined, which causes each chain to fold in a specific way and each protein to assume its particular three-dimensional structure, and this in turn dictates its structural and functional properties.  Inspired by the prowess of nature to build molecules with such precision and heterogeneity, scientists have studied the mechanisms of protein folding - to realize that the dynamics are so complex that predicting a protein's shape given its amino acid sequence is considered to be intractable \cite{fraenkel1993complexity, crescenzi1998complexity}, and engineers have begun to develop artificial systems which fold self-assembling molecules into complex structures~\cite{RotOrigami05, RothOrigami, OrigamiSeed, OrigamiTiles, dill1995} - but with results that to date still lack the diversity of biology.

In order to help progress understanding of the dynamics of systems which self-assemble out of folding components, and to provide a framework for studying such systems, in this paper we introduce the \emph{Flexible Tile Assembly Model} (FTAM).  The FTAM is intended to be a simplified mathematical model of self-assembling systems utilizing components which are able to dynamically reconfigure their relative 3-dimensional locations via folding and unfolding of flexible bonds between components.  It is based on the abstract Tile Assembly Model \cite{Winf98}, and as such the fundamental components are 2-dimensional square \emph{tiles} which bind to each other via \emph{glues} on their edges.  In contrast to the aTAM, in the FTAM each glue type can be specified to either form \emph{rigid} bonds (which force two adjacent tiles bound by such a glue to remain fixed in co-planar positions) or \emph{flexible} bonds (which allow two adjacent tiles bound by such a glue to possibly alternate between being in any of three relative orientations, as shown in Figure~\ref{fig:flexible-positions}).  Because the FTAM is meant to be a test bed for flexible, reconfigurable self-assembling systems, we present a version of the model which makes many simplifying assumptions about allowable positions of tiles and dynamics of the self-assembly process, but which also differs greatly from previously studied self-assembling systems which allow reconfigurability \cite{jRTAM, JonoskaSignals1, JonoskaSignals2, jSignals3D,jSignals,FlexibleCompModel,JonoskaFlexible,CommonUnfoldings, aichholzer2017folding}, other computational studies of folding such as \cite{CommonUnfoldings,aichholzer2017folding}, and algorithmic studies focused on constructing more simple 3D structures such as \cite{KaoFolding}.

\vspace{-15pt}
\begin{figure}[htp]
\centering
    \includegraphics[width=2.0in]{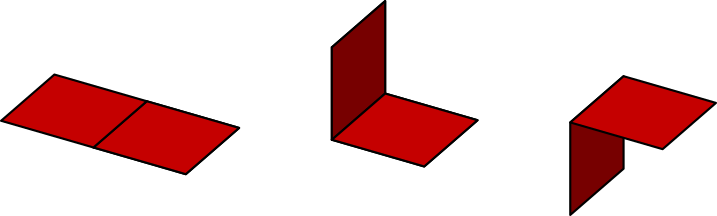}
    \caption{The three relative positions possible for two tiles bound via a flexible glue.}
    \label{fig:flexible-positions}
\end{figure}
\vspace{-10pt}

In Section~\ref{sec:definitions}, we formally introduce the FTAM and provide definitions and algorithms describing its dynamics. In Section~\ref{sec:robustness}, we show how to control flexibility in the model to build 3D shapes. In Section~\ref{sec:utilizing_flexibility}, we present a pair of constructions which demonstrate the potential utility of reconfigurability of assemblies in the FTAM. In the first construction, an FTAM system $\calT$ is given which produces a single terminal assembly that may be in many different configuration. In addition, a set $S$ of $n$ distinct types of tiles are given such that for each subset of $S$, adding this subset of tiles of $S$ to the types of tiles for $\calT$ gives a system with an assembly sequence that starts from the single terminal assembly of $\calT$ and yields a rigid terminal assembly (i.e., an assembly to which no tiles may bind and which, at a high-level, is in a configuration which cannot be folded via flexible glues to give another distinct configuration). Moreover, the resulting rigid assembly is distinct for each choice of subset of $S$. The second construction given in Section~\ref{sec:robustness} demonstrates how a reconfigurable initial assembly can be transformed into either a volume-maximizing hollow cube or a small, tightly compressed brick by selecting between and adding one of two small subsets of tile types. These two constructions demonstrate how algorithmic self-assembling systems could be designed which efficiently (in terms of ``input'' specified by tile type additions) make drastic changes to their surface structures and volumes. These constructions show that FTAM systems can be designed which utilize reconfigurability. In Section~\ref{sec:complexity}, we show that this utility comes at a cost in terms of the computational complexity of determining some important properties of arbitrary FTAM systems. In particular, we show that, given an arbitrary FTAM system, the problem of determining whether it produces an assembly which cannot be reconfigured (via folding along tile edges bonded by flexible glues) is undecidable.  Moreover, we show that, given an assembly, it is co-NP-complete to determine whether the assembly is rigid, i.e. has exactly one valid configuration.  Our final result modifies the previous to show that the problem of deciding if a given assembly for an FTAM system is terminal is also co-NP-complete.  This is especially interesting since, in the aTAM, there is a simple polynomial time algorithm to determine if a given assembly is terminal.

\vspace{-10pt}
\section{Definition of the FTAM} \label{sec:definitions}
\vspace{-5pt}

In this section we present definitions related to the Flexible Tile Assembly Model.

A \emph{tile type} $t$ in the FTAM is defined as a 2D unit square that can be translated, rotated, and reflected throughout 3-dimensional space, but can only occupy a location such that its corners are positioned on four adjacent, coplanar points in $\mathbb{Z}^3$. Each tile type $t$ has four sides $i \in \{N, E, S, W\}$, which we refer to as $t_i$.  Let $\Sigma$ be an alphabet of labels and $\Bar{\Sigma} = \{a^* | a \in \Sigma\}$ be the alphabet of \emph{complementary labels}, then each side of each tile has a \emph{glue} that consists of a \emph{label} $label(t_i) \in \Sigma \cup \Bar{\Sigma} \cup \epsilon$ (where $\epsilon$ is the unique \emph{empty} label for the \emph{null glue}), a non-negative integer \emph{strength} $str(t_i)$, and a boolean valued \emph{flexibility} $flx(t_i)$. (See Figure~\ref{fig:flexible-positions} for a depiction of the positions allowable by a flexible glue.)

A \emph{tile} is an instance of a tile type.  A \emph{placement} of a tile $p = (l,n,o)$ consists of a location $l \in \mathbb{Z}^3$, a \emph{normal} vector $n$ which starts at the center of the tile and points perpendicular to the plane in which the tile lies (i.e. $n \in \{+x,-x,+y,-y,+z,-z\}$\footnote{We refer to the vectors $\{(1,0,0),(-1,0,0),(0,1,0),(0,-1,0),(0,0,1),(0,0,-1)\})$ by the shorthand notation $\{+x,-x,+y,-y,+z,-z\}$ throughout this paper.}), and an \emph{orientation} $o$ which is a vector lying in the same plane as the tile which starts at the center of the tile and points to the $N$ side of the tile (i.e. $o \in \{+x,-x,+y,-y,+z,-z\}$).  Note that by convention, to avoid duplicate location specifiers for a given tile, we restrict a location $l$ to refer to only the $3$ possible tile locations with corners at $l$ and which extend in positive directions from $l$ along one of the planes (i.e. tiles are located by their vertices with the smallest valued coordinates). For any given $l$, there can only be a max of one tile with $n \in \{+x,-x\}$, one tile with $n \in \{+y,-y\}$, and one tile with $n \in \{+z,-z\}$, as to avoid overlapping tiles.

Let $p = (l,n,o)$ and $p' = (l',n',o')$ be placements of tiles $t$ and $t'$, respectively, such that $p$ and $p'$ are non-overlapping\footnote{Non-overlapping placements refer to different tile locations. Formally, two tile placements are non-overlapping if (1) $l != l'$ or (2) $n != n'$ and $n!=\texttt{inverse}(n')$.} and for some $i,j \in \{N,E,S,W\}$, sides $t_i$ and $t_j'$ are adjacent (i.e. touching).  We say that $p$ and $p'$, have \emph{compatible} normal vectors if and only if either (1) $n = n'$, (2) $n$ and $n'$ intersect, or (3) $\texttt{inverse}(n)$ and $\texttt{inverse}(n')$ intersect, where the $\texttt{inverse}$ function simply negates the signs of the non-zero components of a vector. (See Figures~\ref{fig:compatible-normals} and \ref{fig:incompatible-normals}.)  We will refer to these three orientations as ``Straight'', ``Up'', and ``Down'', respectively.  Furthermore, if (1) $label(t_i)$ is complementary to $label(t_j')$, (2) $str(t_i) = str(t_j')$, (3) $flx(t_i) = False$ and $flx(t_j') = False$, and (4) $n$ and $n'$ are in a ``Straight'' orientation, then the glues on $t_i$ and $t_j'$ can \emph{bind} with strength value $str(t_i)$ to form a \emph{rigid bond}. Similarly, if (1) $label(t_i)$ is complementary to $label(t_j')$, (2) $str(t_i) = str(t_j')$, (3) $flx(t_i) = True$ and $flx(t_j') = True$, and (4) $n$ and $n'$ are compatible, then the glues on $t_i$ and $t_j'$ can \emph{bind} with strength value $str(t_i)$ to form a \emph{flexible bond}. \footnote{Note that any glue can only bind to a single other glue. Also, we do not allow two pairs of coplanar tiles to bind through the same space (i.e. the two partial surfaces created by two pairs of bounded coplanar tiles are not allowed to intersect). Therefore, $4$ glues from $4$ different tiles that are all adjacent to each other can all form bonds only if they form two flexible bonds in non-straight orientations.}

\vspace{-15pt}
\begin{figure}[htp]
\centering
    \subfloat[][Compatible normal vectors.]{
        \label{fig:compatible-normals}
        \includegraphics[width=1.8in]{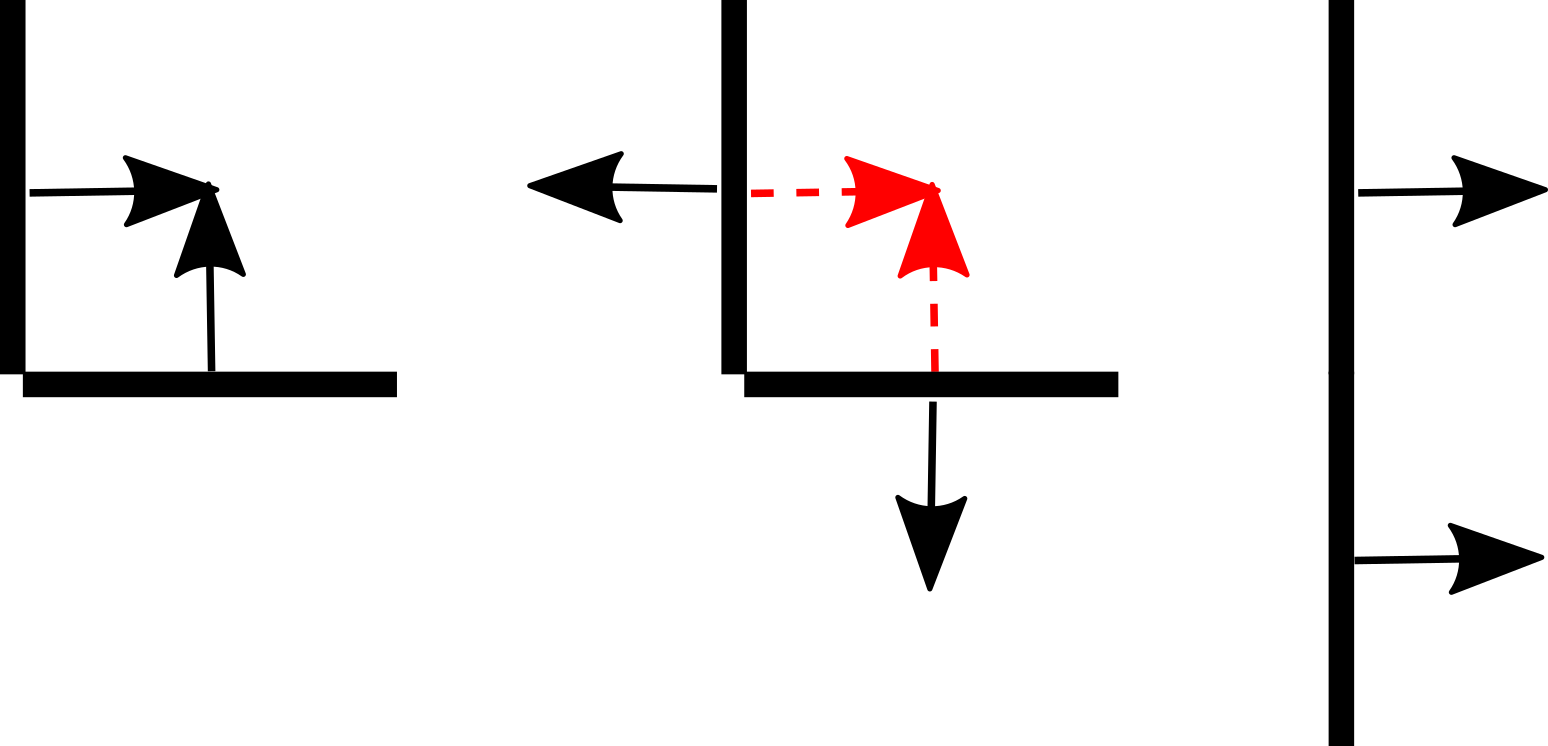}
    }
    \quad\quad
    \subfloat[][Incompatible normal vectors.]{
        \label{fig:incompatible-normals}
        \includegraphics[width=1.8in]{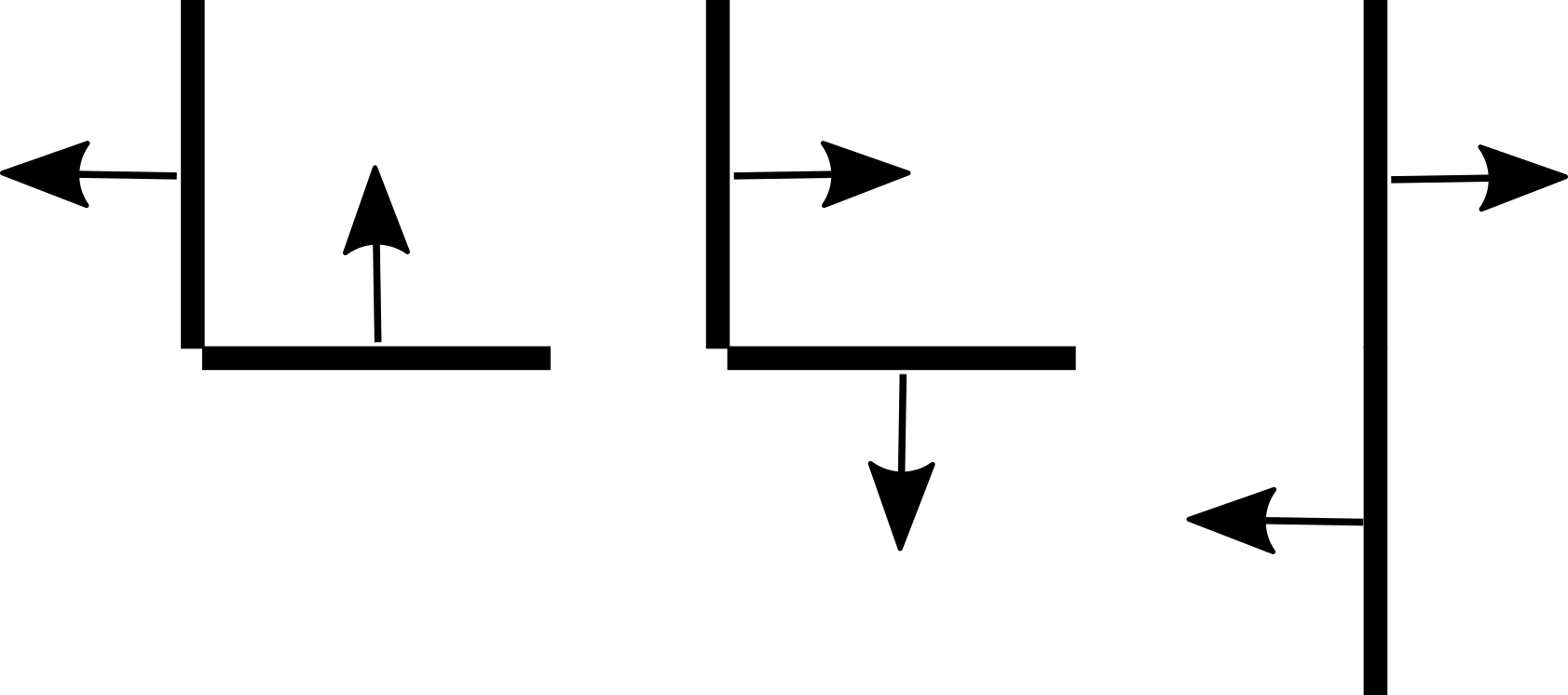}
    }
    \caption{Possible normal vectors of pairs of tiles.  Those in (a) are compatible and allow a bond to form between complementary glues in the orientations ``Up'', ``Down'', and ``Straight'', respectively. Those in (b) are not compatible.}
    \label{fig:normals}
\end{figure} \hfill

We define an \emph{assembly} $\alpha$ as a graph whose nodes, denoted $V(\alpha)$, are tiles and whose edges, denoted $E(\alpha)$, represent bound complementary glues between adjacent edges of two tiles. An edge between sides $i$ and $j$ of tiles $t$ and $t'$, respectively, is represented by the tuple $(t_i,t_j')$, which specifies which sides of $t$ and $t'$ the bond is between.  Whether it is flexible is denoted by $flx(t_i)$ and its strength is denoted by $str(t_i)$ (since those values must be equal for $t_i$ and $t_j'$).

We define a \emph{face} to be a set of coplanar tiles that are all bound together through rigid bonds. Additionally, we define a \emph{face graph} to be a graph minor of the assembly graph where every maximal subgraph in which every node can be reached from every other node using a path of rigid tiles is replaced by a single node in the face graph. Two nodes in the face graph that correspond to two groups of nodes in the assembly graph have an edge if and only if there is at least one flexible bond between any single node in the first group of the assembly graph and any single node in the second group of the assembly graph.

An FTAM system is a triple $\mathcal{T} = (T,\sigma,\tau)$ where $T$ is a finite set of tile types (i.e. tile set), $\sigma$ is an initial \emph{seed} assembly, and $\tau \in \mathbb{Z}^+$ is a positive integer called \emph{temperature} which specifies the minimum binding threshold for tiles. An assembly is $\tau$-stable if and only if every cut of edges of $\alpha$ which separates $\alpha$ into two or more components must cut edges whose strengths sum to $\ge \tau$.  We will only consider assemblies which are $\tau$-stable (for a given $\tau$), and we use the term assembly to refer to a $\tau$-stable assembly.

Given an assembly $\alpha$, a \emph{configuration} $c_\alpha$ is a mapping from every flexible bond in $\alpha$ to an orientation from \{``Up'', ``Down'', ``Straight''\}. An \emph{embedding} $e_\alpha$ is a mapping from each tile in $\alpha$ to a placement. Given an assembly and a configuration, we can obtain an embedding by choosing any single initial tile and assigning it a placement and computing the placement of each additional tile according to how it is bonded with tiles that are already placed. Note that, given tiles to which it is bound, their placements, and an orientation, there is only one tile location at which each additional tile can be placed. We say a configuration $c_{\alpha}$ is \emph{valid} if and only if an embedding obtained from the configuration (1) does not place more than one tile at any tile location, (2) doesn't bond tiles through the same space, and (3) does not have contradicting bond loops. To elaborate on (2), while 4 glues can all be adjacent at one point, we allow them to bind in pairs in ``Up'' or ``Down'' orientations but do not allow both pairs to bind across the gap in ``Straight'' orientations. To elaborate on (3), contradicting bond loops occur when placing a series of tiles that are all bound in a loop causes the last tile to be placed at a location that is not adjacent to the first tile, therefore making the loop unable to close. Examples of configurations that follow and contradict (3) are given in Figure \ref{fig:validation}. Note that two embeddings that use different initial tiles and initial placements but the same configuration will be equivalent up to rotation and translation.

\begin{figure}
\centering
    \includegraphics[width=0.65\textwidth]{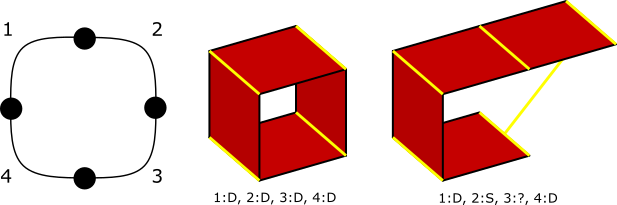}
    \caption{Here we see an assembly, a valid configuration, and an invalid configuration. In the third image, because of the orientations of bonds 1, 2, and 4, bond 3 is between two tiles that are not connected, making the configuration invalid.}
    \label{fig:validation}
\end{figure}

Let $\alpha$ be an assembly and $c_\alpha$ and $c_\alpha'$ be valid configurations of $\alpha$.  If for every flexible bond $b \in \alpha$ either $c_\alpha(b) = Up$ and $c_\alpha'(b) = Down$, $c_\alpha(b) = Down$ and $c_\alpha'(b) = Up$, or $c_\alpha(b) = Straight$ and $c_\alpha'(b) = Straight$, we say that $c_\alpha$ is the \emph{chiral} configuration of $c_\alpha'$ and vice versa. Note that the embeddings achieved from $c_\alpha$ and $c_\alpha'$ are reflections of each other. We refer to the special reconfiguration of an assembly to its chiral as \emph{inversion}.

Given an assembly $\alpha$ and two different embeddings $e_\alpha$ and $e_\alpha'$, we say that $e_\alpha$ and $e_\alpha'$ are \emph{equivalent}, written $e_\alpha \equiv e_\alpha'$, if one can be rotated and/or translated into the other. If two embeddings are equivalent, this means they were computed from the same configuration, although possibly using a different placement for the initial tile.

We define the set of all valid configurations of $\alpha$ as $\mathcal{C}(\alpha)$.  We say that an assembly $\alpha$ is \emph{rigid} if (1) $|\mathcal{C}(\alpha)| = 1$, or (2) $|\mathcal{C}(\alpha)| = 2$ and the two valid configurations are chiral versions of each other.  Conversely, if $\alpha$ is not rigid, we say that it is \emph{flexible}.

The \emph{frontier} of a configuration $c_\alpha$, denoted $\partial^{\mathcal{T}} c_\alpha$, is the set composed of all pairs $(t,B)$ where $t \in T$ is a tile type from tile set $T$ and $B$ is a set of up to 4 tile/glue pairs such that an embedding of $c_\alpha$ would place each tile adjacent to one location such that a tile of type $t$ could bind to each glue for a collective strength greater than or equal to the temperature parameter $\tau$. Given an assembly $\alpha$ and a set of valid configurations $\mathcal{C}(\alpha)$, we define the multiset of frontier locations of assembly $\alpha$ across all valid configurations to be $\hat{\partial}^{\mathcal{T}}\alpha = \bigcup_{c_\alpha \in \mathcal{C}(\alpha)} \partial^{\mathcal{T}} c_\alpha$, i.e. $\hat{\partial}^{\mathcal{T}}\alpha$ is the multiset resulting from the union of the sets of frontier locations of all valid configurations of $\alpha$.

Given assembly $\alpha$ and valid configuration $c_\alpha$, $\#(c_\alpha)$ is the maximum number of new bonds which can be formed across adjacent tile edges in an embedding of $\alpha$ which are not already bound in $\alpha$ (i.e. these are tile edges which have been put into placements allowing bonding in configuration $c_\alpha$ but whose bonds are not included in $\alpha$).  We then define $\mathcal{C}_{max}(\alpha) = \{ c_\alpha | c_\alpha \in \mathcal{C}(\alpha)$ and $\forall c_\alpha' \in \mathcal{C}(\alpha)$, $\#(c_\alpha) \ge \#(c_\alpha')\}$. Namely, $\mathcal{C}_{max}(\alpha)$ is the set of valid configurations of $\alpha$ in which the maximum number of bonds can be immediately formed.

Given an assembly $\alpha$ in FTAM system $\mathcal{T}$, a single step of the assembly process intuitively proceeds by first randomly selecting a frontier location from among all frontier locations over all valid configurations of $\alpha$. Then, a tile is attached at that location to form a new assembly $\alpha'$. Next, over all valid configurations of $\alpha'$, a configuration is randomly selected in which the maximum number of additional new bonds can be formed (i.e. the addition of the new tile may allow for additional bonds to form in alternate configurations, and a configuration which maximizes these is chosen), and all possible new bonds are formed in that configuration, yielding assembly $\alpha''$.
Assuming that $\alpha$ was not \emph{terminal} and thus $\alpha'' \neq \alpha$, we denote the single-tile addition as $\alpha \to_1^{\mathcal{T}} \alpha''$.  To denote an arbitrary number of assembly steps, we use $\alpha \to_*^{\mathcal{T}} \alpha''$.  For an FTAM system $\mathcal{T} = (T,\sigma,\tau)$, assembly begins from $\sigma$ and proceeds by adding a single tile at a time until the assembly is terminal (possibly in the limit).
(See Section~\ref{sec:definitions-append}
 for pseudocode of the assembly algorithms.)
For any $\alpha'$ such that $\sigma \to_*^{\mathcal{T}} \alpha'$, we say that $\alpha'$ is a \emph{producible} assembly and we denote the set of producible assemblies as $\prodasm{\mathcal{T}}$.  We denote the set of terminal assemblies as $\termasm{\mathcal{T}}$.

Note that in this section we have provided what is intended to be an intuitively simple version of the FTAM in which the full spectrum of all possible configurations of an assembly are virtually explored at each step, and only those which maximize the number of bonds formed at every step are selected. Logically, this provides a model in which assemblies reconfigure into globally optimal configurations, in terms of bond formation, between each addition of a new tile. Clearly, depending on the size of an assembly and the degrees of freedom of various components afforded by flexible bonds, such optimal reconfiguration could conceivably be precluded by faster rates of tile attachments.  Various parameters which seek to balance the amount of configuration-space exploration versus tile attachment rates have been developed to study more kinetically realistic dynamics, but are beyond the scope of this paper.

\section{Controlling Flexibility to Build Structures} \label{sec:robustness}

Our goal in this section is to deterministically assemble certain shapes in the FTAM at temperature two. We define a \emph{shape} to be a collection of connected tile locations. A shape is invariant through translation and rotation. Rather than go through an endless case-by-case analysis of all possible shapes, we focus on collections of 2D tile locations that form the outlines of three-dimensional shapes. We refer to these 3D shapes as \emph{polycubes} and the set of 2D tile locations on their outer surface as an \emph{outline}. We say that an FTAM system $\mathcal{T} = (T,\sigma,\tau)$ \emph{deterministically assembles} a shape $s$ if the embedding of all configurations $\mathcal{C}_\alpha$ of all terminal assemblies $\termasm{\mathcal{T}}$ of the system $T$ have shape $s$.

Due to the definition of the model, the most prominent additional challenge that is present in FTAM systems over traditional 2D aTAM systems is controlling the orientation of different faces in the assembly relative to one another as the assembly process is occurring. In which case, the approach that we use to demonstrate shape building in the FTAM is to make an \emph{edge frame} for each polycube using unique tile types and filling in each face. We define an edge frame to be the collection of the outer-most tiles of each face in the outline of a polycube. For now, we will make the assumption that every edge of the shape is connected and will address this later in the section. We claim that studying edge frames is sufficient for unveiling the power of the FTAM to orient new faces in the assembly process since, intuitively, the cooperation of other tiles on the edges of adjacent faces doesn't provide additional help in correctly orienting those faces over just the tiles at the vertex. This intuition stems from the idea that the faces of a shape incident on a vertex interact on the same axes that the individual tiles incident on a vertex do.

One big deciding factor about whether the outline of a specific polycube can be made in the FTAM comes down to the types of vertices in that polycube. Because of this, we continue our analysis by breaking down the types of vertices that can exist on a polycube. Every type of vertex possible on a polycube can be enumerated by enumerating all polycubes that can fit inside a $2\times2\times2$ space that are distinct up to rotation and reflection. You can see the outcome of this enumeration in Figure \ref{fig:possible_outline_vertices}. In each polycube, the vertex type is illustrated at the center point of the $2\times2\times2$ space. The illustration has labels to later reference each vertex type.

\begin{figure}[htp]
\centering
    \includegraphics[width=1.0\textwidth]{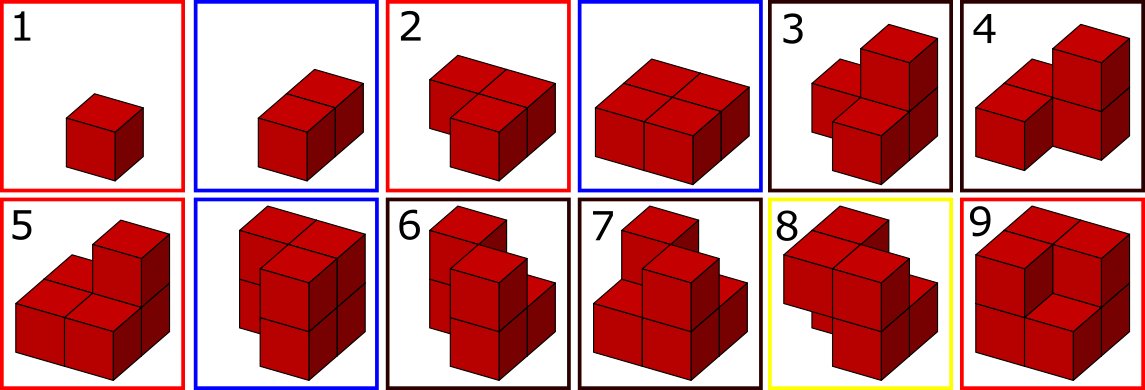}
    \caption{All possible polycubes that can fit inside of a $2\times2\times2$ space, and furthermore, all possible vertex types that could exist on a polycube.}
    \label{fig:possible_outline_vertices}
\end{figure}

This yields 6 distinct vertex types. Vertices 1, 8, and 9 are all the same, which we refer to as a \emph{convex} vertex. Vertices 3 and 5 are the same, which we refer to as a \emph{concave} vertex. Vertex 3, 4, 6, and 7 are all distinct and we refer to them collectively as the \emph{complex} vertices. In addition to the vertex type, the system must also be able to deterministically assemble the vertex from the correct perspective. A \emph{perspective} is the relative direction that the new edges that form with the vertex are pointing with respect to the tiles of the original edge that first grows up to the vertex in the assembly process. Each vertex can have any number of perspectives from 1 (\emph{symmetric} vertex) to the number of edges (\emph{asymmetric} vertex), inclusive.

Collectively, there are 15 unique perspectives among the 6 distinct vertices. Each perspective requires its own tiling protocol to get the vertex to configure correctly. We construct these protocols using (a) the number of tiles that are incident on the vertex that are bound in a loop (which we refer to as the \emph{loop length}) and (b) the sequence of flexibility values in the bonds of the loop (which we refer to as the \emph{bond sequence}). If a perspective has a unique protocol, then attaching a loop of tiles using the protocol will result in the only possible configuration available for the loop of tiles being the correct perspective. Of the 15 perspectives previously mentioned, 11 perspectives among 4 of the vertices will have their own unique protocol and will therefore be deterministic. The other 4 perspectives among 2 vertices will all share one protocol and will therefore not be deterministic. These 2 vertices are types 3 and 7, which we will subsequently refer to as \emph{reconfigurable} vertices. For a full enumeration and discussion of vertices, perspectives, and tiling protocols, see
Section \ref{sec:verts_and_perts}.

\begin{figure}[htp]
\centering
    \subfloat{%
        \label{fig:vertex_control_setup}%
        \includegraphics[width=0.24\textwidth]{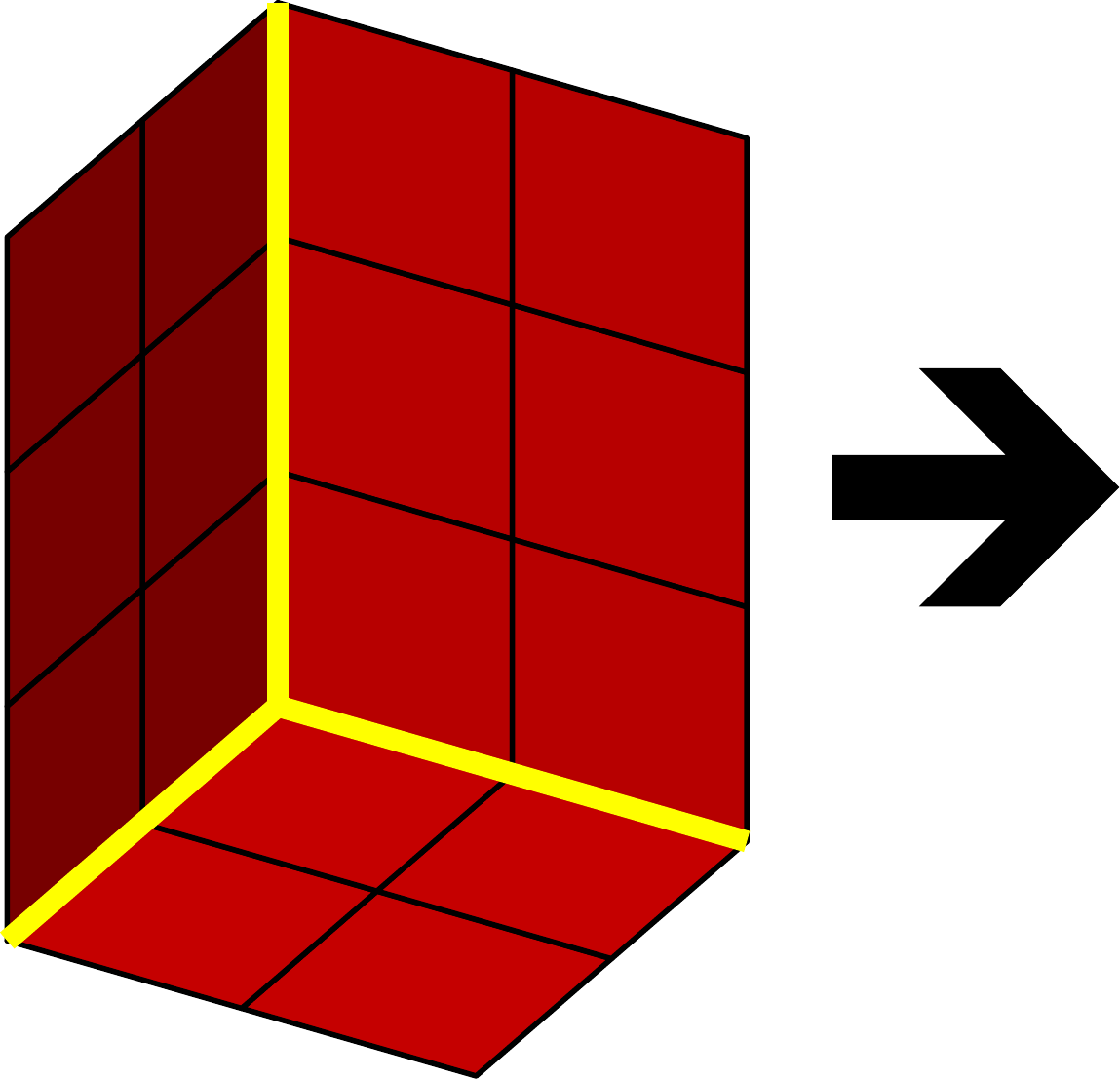}
        }
  \subfloat{%
        \label{fig:convex_vertex}%
        \includegraphics[width=0.16\textwidth]{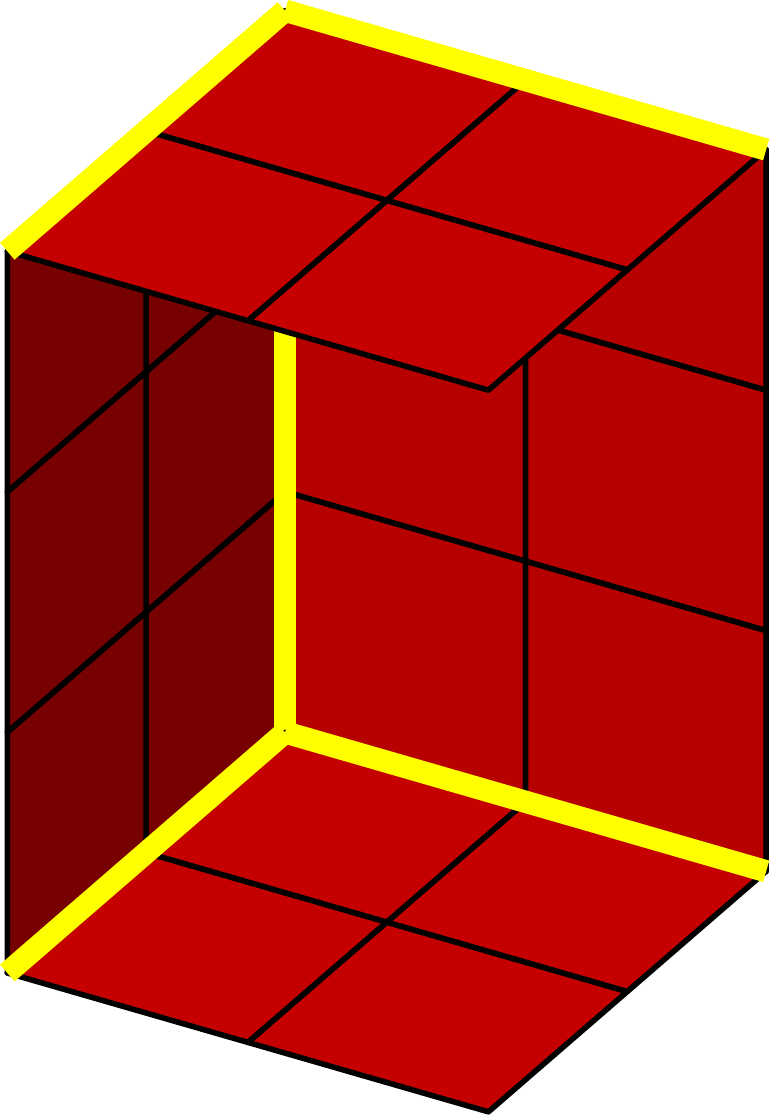}
        }
  \subfloat{%
        \label{fig:concave_vertex}%
        \includegraphics[width=0.30\textwidth]{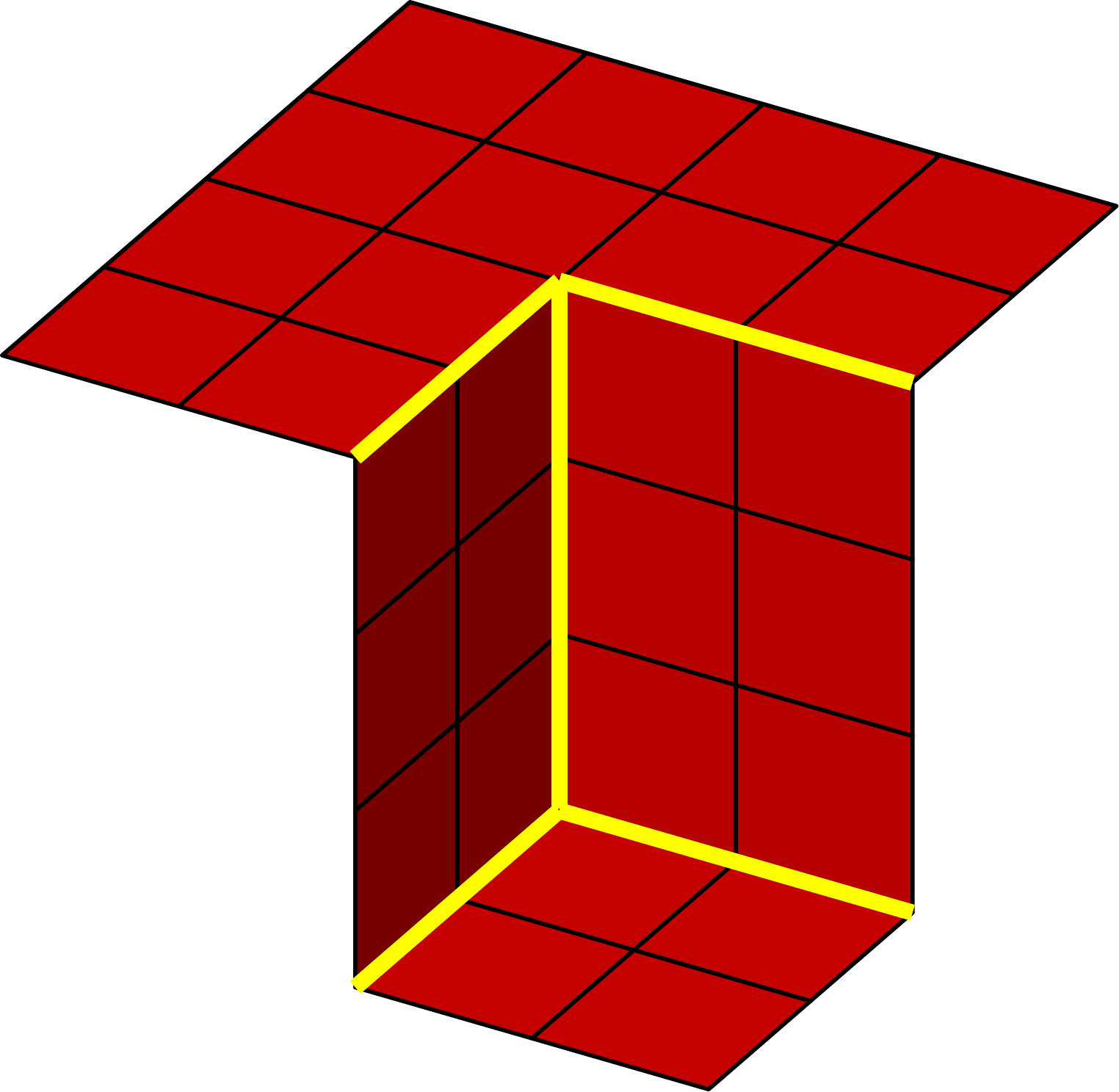}
        }
  \subfloat{%
        \label{fig:combined_vertex}%
        \includegraphics[width=0.20\textwidth]{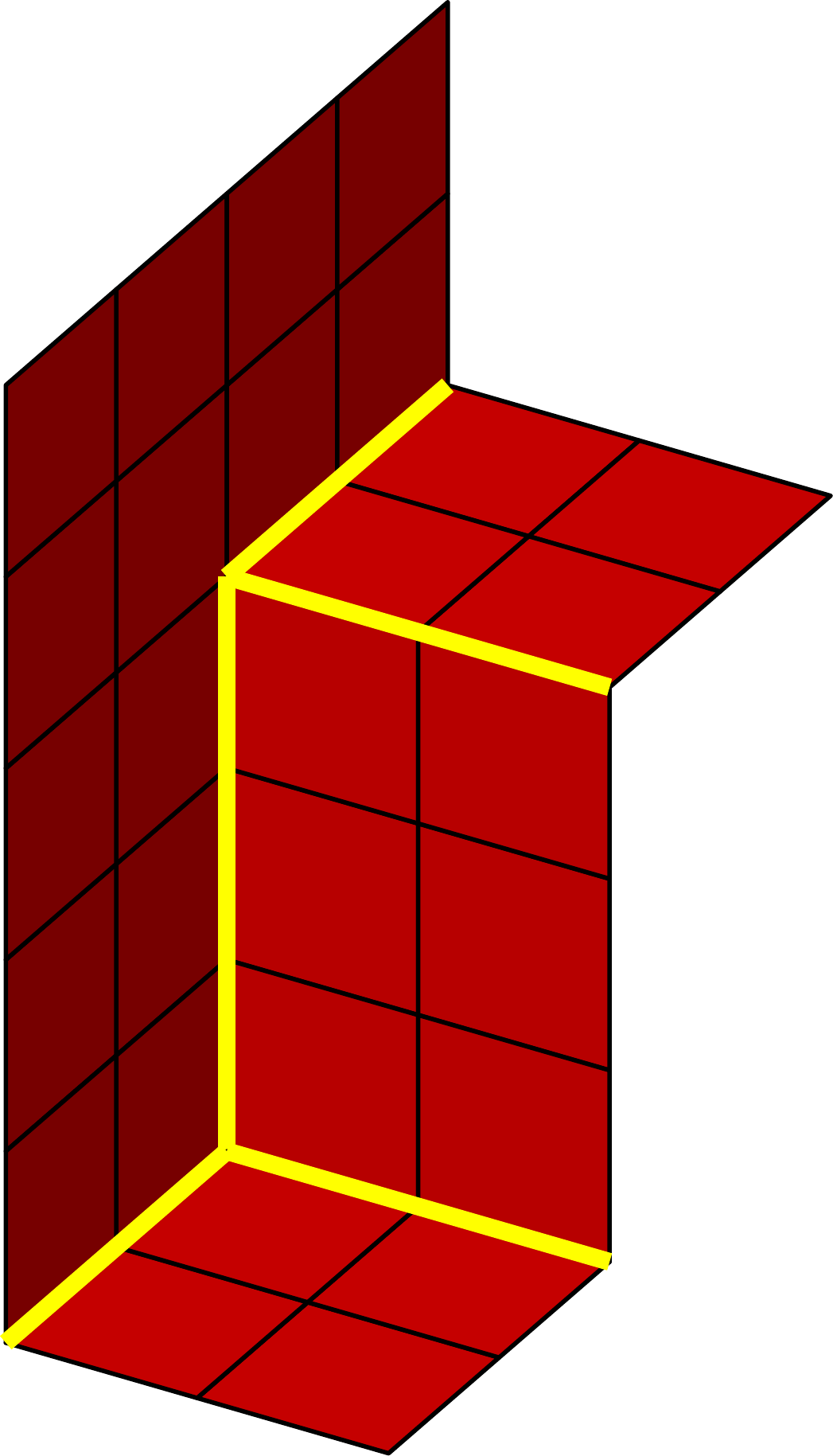}
        }%
        \quad
        
  \caption{(a) original edge, (b) convex vertex, (c) concave vertex, from one unique perspective, (d) concave vertex, from another perspective}
  \label{fig:types_of_vertices}
\end{figure}

\paragraph{Assembly Process.}

Now, we consider the assembly process. Let's assume we start with a seed that is just the three tiles in a simple convex vertex. Notice that as the assembly process starts, the seed vertex and the edges that are growing out from it can invert as a whole but cannot otherwise reconfigure (since that would require removing a bond from the assembly). (For assembling an edge, we outline a trivial protocol in
Section~\ref{sec:edge_protocol}
.) Each time the assembly grows up to a vertex, it will attach the loop of tiles that make this new vertex. As long as the new vertex is not a reconfigurable vertex, it will be forced to take a configuration that agrees with configuration of the seed vertex. By this, we mean that, if the seed vertex were to invert at this point, the edge connecting the two vertices would invert, and the new vertex would therefore be forced to invert. This cause-effect relationship is true for any vertices (excluding reconfigurable vertices) connected by an edge, which means that, if any bond in the partial assembly were to reorient, the whole partial assembly must invert, i.e. inversion is the only possible reconfiguration. An example of an edge frame started from a potential seed is shown in Figure \ref{fig:edge_frame_example}.

\begin{figure}[htp]
\centering
    \includegraphics[width=\textwidth]{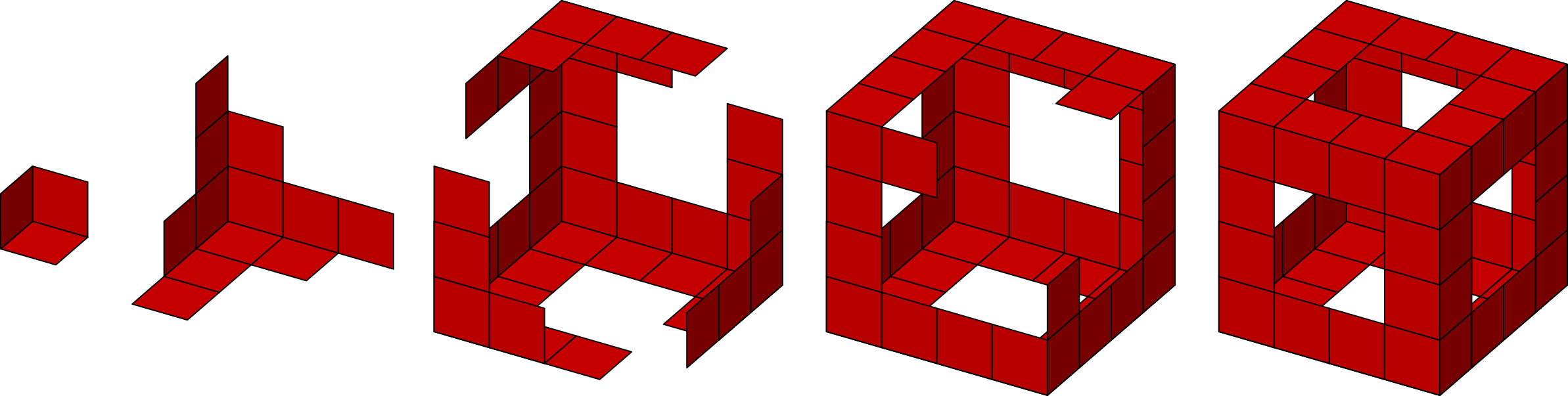}
    \caption{An assembling edge frame starting from a potential seed. Each edge grows up to a vertex and initializes other edges until the whole frame has filled out.}
    \label{fig:edge_frame_example}
\end{figure}

We now prove a claim that assembling in the correct configuration or the chiral configuration is identical (since both configurations have the same frontier) and will therefore yield the same shape.

\begin{claim}
Every frontier location $f$ in an assembly $\alpha$ for a given configuration $c_\alpha$ has a corresponding frontier location $f'$ in $\alpha$ in the chiral configuration $c_\alpha'$, such that attaching $f$ to $\alpha$ in $c_\alpha$ produces the same assembly but in the chiral configuration of attaching $f'$ to $\alpha$ in $c_\alpha'$.
\end{claim}

\begin{proof}
Notice that a frontier location in the FTAM is dependent on 12 neighboring tile locations, an ``Up'', ``Straight'', and ``Down'' location for each of the 4 sides of the tile. Also remember chiral configurations of an assembly $\alpha$ produce embeddings of $\alpha$ that are the reflections of each other. Now, take any frontier location $f$ in $c_\alpha$. By reflecting an embedding of $c_\alpha$ over the plane that $f$ exists in, the 12 tile locations that make $f$ into a frontier location will still be neighboring $f$, with the ``Up'' and ``Down'' neighboring locations switching places and also reflecting, thereby keeping the same glues incident on the location of $f$. Since all the same glues are incident on the tile location, this location, which we will call $f'$, is also a frontier location in $c_\alpha'$ with the same tile type as in $c_\alpha$, even if $c_\alpha'$ includes some translation or rotation. Since the frontier locations are on the plane of symmetry that we used to get the chiral configurations, adding the tile to the assembly in either configuration will produce two configurations that are also chiral configurations of each other.
\end{proof}

Once the assembly process has finished, the terminal assembly could also flip between the correct shape in its chiral. When there is at least one plane of symmetry in the shape, then reconfiguration in the assembly process actually will not prevent the system from being deterministic. This is because the chiral of a symmetric polycube is itself. Therefore, although the system will technically make two different terminal assemblies, one can be rotated into the other, meaning that the two different terminal assemblies have the same shape by definition, making the system deterministic.

\paragraph{Multiple Edge Frames.}

Up to this point, we have assumed all the edges in a polycube are connected. However, this is not always the case. For example, anytime two pieces of a shape are connected by a set of coplanar tiles (i.e. when the face graph has a cut vertex). Shapes like this are a problem because they require multiple edge frames to build, and similar to the chirality of asymmetric shapes, additional edge frames can also have chiral reconfigurations. Therefore, disagreeing chiralities of the edge frames can configure the terminal assembly of a system into a shape that is neither the intended shape nor its chiral. In general, each additional edge frame doubles the number of configurations that the terminal assembly can exist in, only one of which (or two, if symmetric) is the desired shape. There are some exceptions to this (as discussed in
Section~\ref{sec:combating_multiple_edge_frames})
such as blocking and symmetry.

\paragraph{Summary.}

Combining the results of this section, we get the following theorem.

\begin{theorem}
A temperature two FTAM system can deterministically assemble the outline of any polycube that meets the following conditions:

\begin{enumerate}
    \item the polycube is symmetric,
    \item there are no reconfigurable vertices in the polycube, and
    \item the edges of the polycube are all connected
\end{enumerate}
\end{theorem}

\section{Utilizing Flexibility} \label{sec:utilizing_flexibility}

As discussed previously, reconfigurability may be able to provide assembly systems with interesting properties that enable diverse applications. For example, changing geometry on the surface of a synthetic structure may allow it to interact with varying other structures in a system, or contracting/expanding volumes may impact how well it can diffuse through narrow channels. With a simple extension to the base FTAM model which allows an initial terminal assembly to form, and then at a later stage the addition of a new set of tile types allows the assembly to reconfigure, an assembly's final shape can be locked in based on these additional tiles. As previously mentioned, we extend the FTAM here to allow such staged assembly as the simplest mechanism for leveraging this type of reconfigurability, but note that alternative mechanisms could also work, such as glue activation and deactivation \cite{Signals}.

\subsection{Staged Functional Surface: Maximizing the number of reconfigurations}

For our first demonstration of a construction utilizing flexibility as a tool, we present a construction which maximizes the number of rigid configurations which a flexible assembly (formed during a first stage of assembly) can be locked into, based on the number of new tile types added during a second stage of assembly.  Figure \ref{fig:staged_functional_surface} gives a high-level schematic of a simple example of such a system. (Note that we omit full details of each tile type as these components can all be easily constructed using standard aTAM techniques and techniques from Section~\ref{sec:robustness}.) It shows the inner-makings of an initial structure that can later be modified by adding new tiles types into solution. We refer to this structure as a \emph{film}. The film works by allowing the tiles in the very top layer to move freely. By adding select subsets of tile types during the second stage, prescribed tiles can be pinned up from the surface or pinned to the bottom layer of the assembly. Pinning up works by using the second layer of the film (from the bottom) to block the incoming tiles from folding down into the assembly, thereby forcing them to fold up. Pinning down works by connecting the top layer to the bottom layer of the film, forcing the tiles to fold down. The bumps formed from pinning up, also called \emph{pixels}, can be arranged into a specified geometry, or \emph{image}.
 The setup of this system is shown in Figure \ref{fig:staged_functional_surface_example}.
The eight tiles on top are used to pin up the pixels in the image. The other tiles specified on the bottom can be used to pin down the rest of the free pieces in the assembly if this is required instead.

For this system, if the side lengths of the film are $n$, note that there are $O(n^2)$ potential pixel locations, meaning that there are a maximum of $O(2^{n^2})$ possible pixel configurations (i.e. each can be either up or down in any given configuration).  To transform the flexible film into a rigid configuration with a particular set of pixels projecting upward, it is necessary to add tiles of $O(n^2)$ tile types corresponding to the up or down orientations, which is optimal as each tile type is encoded by a constant number of bits and $\log(O(2^{n^2})) = O(n^2)$ bits are necessary to uniquely identify each of the $O(2^{n^2})$ configurations.  Note that although these reconfigurations are relatively trivial, the differences in the sizes of the reconfigurable sections can be arbitrarily large without requiring more unique tile types to be added in the second stage.  This construction displays a maximum number of resulting rigid configurations from an optimal number of additional tile types in the second stage.

\vspace{-10pt}
\begin{figure}[htp]
\centering
    \includegraphics[width=0.9\textwidth]{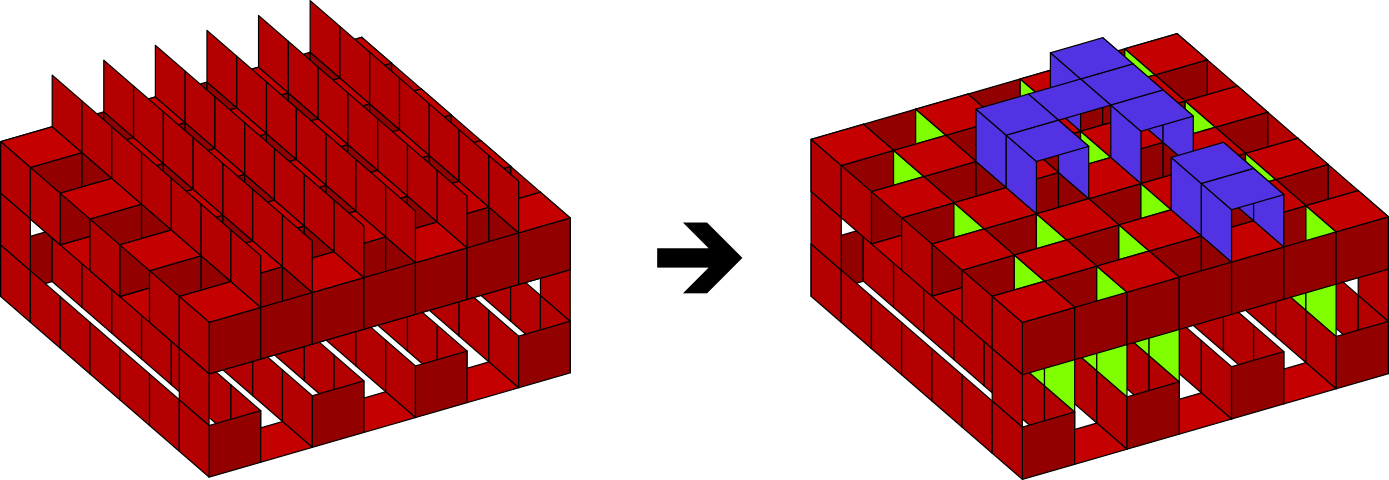}
    \caption{An example of reconfigurable shape that can be used in a staged environment to display a functional surface only after additional tiles are added.}
    \label{fig:staged_functional_surface}
\end{figure}
\vspace{-25pt}

\subsection{Compressing/expanding structures}\label{sec:compress-expand}

We now demonstrate a construction that is able to take advantage of the flexibility of bonds in the FTAM to allow a base assembly to lock into an expansive, rigid but hollow configuration given the addition of one subset of tile types in the second stage, or to instead lock into a compressed, compact and dense configuration given the addition of a different subset of tile types.

\begin{figure}[htp]
\centering
    \includegraphics[width=\textwidth]{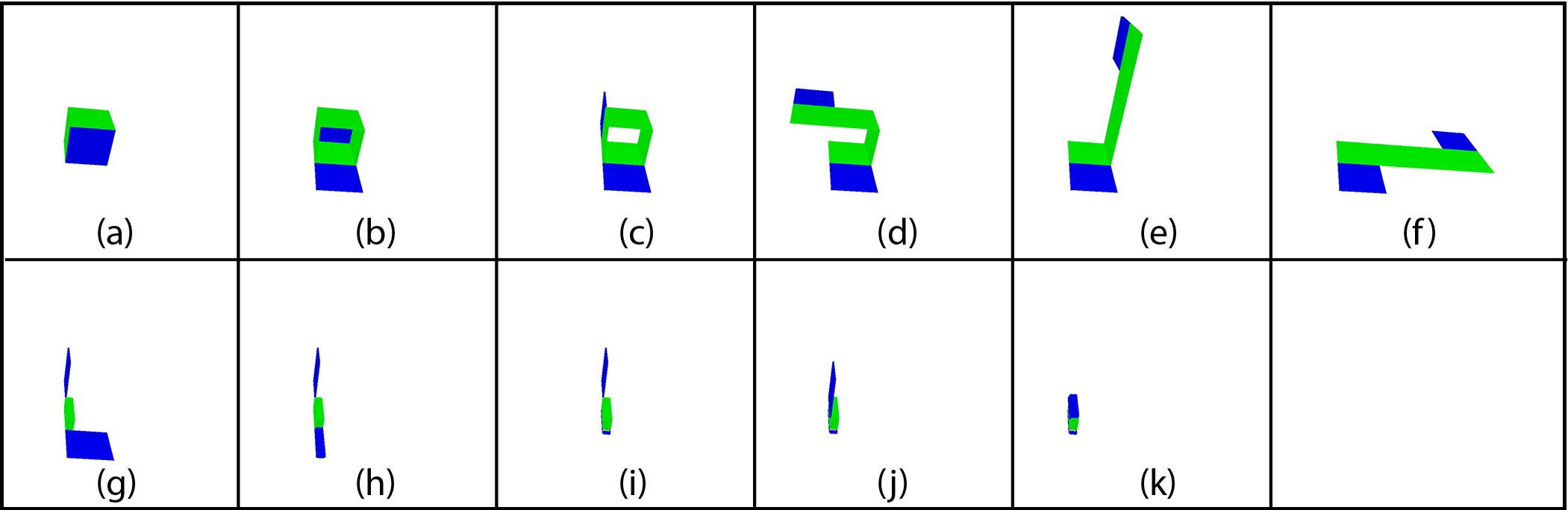}
    \caption{Series of images giving a schematic depiction of a transformation from a hollow $n \times n \times n$ cube into a compressed, approximately $n \times \sqrt{n} \times \sqrt{n}$ configuration. \vspace{10pt}}
    \label{fig:compression-animation}
\end{figure}

Figure~\ref{fig:compression-animation}(f) shows an assembly of six approximately $n \times n$ squares attached together and in a flattened ``sheet'' configuration.  Such an assembly can be efficiently self-assembled using $O(\log(n))$ tile types in the first stage. In the second stage, one of two sets of a constant number of additional tile types could be added so that either (1) the sheet folds into a hollow, volume-maximizing cube of dimensions $n \times n \times n$ (i.e. volume $n^3$).  A schematic representation of the transformation can be seen going backward from Figure~\ref{fig:compression-animation}(f) to Figure~\ref{fig:compression-animation}(a), or (2) the sheet folds into a compressed, compact ``brick'' of dimensions $O(n) \times O(\sqrt{n}) \times O(\sqrt{n})$ (i.e. volume $n^2$).  A schematic representation of the transformation can be seen going forward from Figure~\ref{fig:compression-animation}(f) to Figure~\ref{fig:compression-animation}(k).

\section{Complexity of FTAM Properties} \label{sec:complexity}
\vspace{-5pt}

In this section we consider the problem of deciding if a system produces rigid assemblies, the problem of deciding for a given assembly, if the assembly is rigid, and the problem of deciding for a given assembly, if the assembly is terminal. We first consider the problem of deciding if a system produces rigid assemblies.

\vspace{-10pt}
\subsection{Determining if a system produces a rigid assembly is uncomputable}
\vspace{-5pt}

We first show that, given an arbitrary FTAM system, determining if it produces a rigid terminal assembly is undecidable.

\begin{problem}[Rigidity-from-system]
Given an FTAM system $\mathcal{T}$, does there exist assembly $\alpha \in \termasm{\mathcal{T}}$ such that $\alpha$ is rigid?
\end{problem}

\begin{theorem} \label{thm:Rigidity}
Rigidity-from-system is undecidable.
\end{theorem}
\vspace{-5pt}

For any given Turing machine $M$, we show that $M$ can be simulated by an FTAM system that produces a single terminal rigid assembly iff $M$ halts. The full proof of Theorem~\ref{thm:Rigidity} is given in
Section~\ref{sec:complexity_appendix}. 

\vspace{-10pt}
\subsection{Determining the rigidity of an assembly is co-NP-complete}\label{sec:assembly-rigidity}
\vspace{-5pt}

Now, we look at the complexity of determining the rigidity of a given assembly. %

\begin{problem}[Rigidity-from-assembly]
Given an FTAM system $\mathcal{T}$ and assembly $\alpha \in \prodasm{\mathcal{T}}$, is $\alpha$ rigid?
\end{problem}

\begin{theorem}\label{thm:rigidity-from-assembly}
  Rigidity-from-assembly is co-NP-complete.
\end{theorem}

To prove Theorem~\ref{thm:rigidity-from-assembly}, we prove the following two lemmas.

\begin{lemma}\label{lem:rigidity-from-assembly-co-NP}
  The complement of rigidity-from-assembly is in NP.
\end{lemma}

\begin{proof}
To illustrate this, we take an instance of the problem that contains the FTAM system $\mathcal{T}$ and assembly $\alpha \in \prodasm{\mathcal{T}}$. Our certificate in this instance will be configurations $c_\alpha$ and $c_\alpha'$. Since a configuration is simply a mapping from every flexible bond in $\alpha$ to an orientation, each configuration requires $O(|\alpha|)$ space, and thus the cerficate is polynomial in the size of $\alpha$.  To determine if the certificate is valid, and thus if $\alpha$ is flexible (and therefore not rigid), we first check that $c_\alpha$ and $c_\alpha'$ are valid encodings of a configurations, meaning they each map every flexible bond in $\alpha$ to an orientation from \{``Up'', ``Down'', ``Straight''\}. Then we must ensure that $c_\alpha'$ is different than $c_\alpha$. Both of these can be done in linear time with respect to the number of flexible bonds in the assembly. Next, we compute embeddings of $\alpha$ from $c_\alpha$ and $c_\alpha'$, taking linear time in the number of tiles in the assembly. While computing the embeddings, we simply check that no tile is assigned a placement already taken by another tile, that no bonds overlap the same space, and that every tile is adjacent to the tiles it is connected to in $\alpha$ such that their glues line up correctly. Computing the embeddings and checking these conditions takes linear time with respect to the number of tiles in the assembly. If all of these conditions are met, then both $c_\alpha$ and $c_\alpha'$ are valid configurations of $\alpha$, and therefore $\alpha$ is not rigid.  Since the certificate has polynomial size in relation to $\alpha$ and can be verified in polynomial time to show that $\alpha$ is not rigid, the problem of determining if $\alpha$ is rigid is in co-NP.
\end{proof}

\vspace{-10pt}
\begin{lemma}\label{lem:rigidity-from-assembly-NP-hard}
  The complement of rigidity-from-assembly is NP-hard.
\end{lemma}
\vspace{-5pt}

We prove Lemma~\ref{lem:rigidity-from-assembly-NP-hard} by a 3SAT reduction. In particular, we give an FTAM system, $\mathcal{T}$ say, and show how to encode a 3SAT formula as a producible assembly of $\mathcal{T}$ in a configuration, $c$ say, such that there exists a configuration $c'$ of $\alpha$ that is distinct from $c$ iff the 3SAT formula is satisfiable.
(See Section~\ref{sec:complexity} for details.)

Finally, Theorem~\ref{thm:rigidity-from-assembly} is proven by Lemmas~\ref{lem:rigidity-from-assembly-co-NP} and \ref{lem:rigidity-from-assembly-NP-hard}.

\vspace{-15pt}
\begin{figure}[htp]
\centering
    \subfloat{%
        \includegraphics[width=0.240\textwidth]{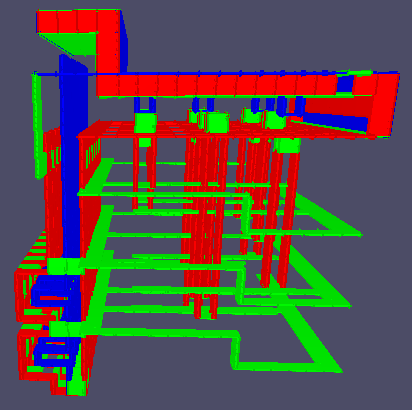}}
        \enskip
    \subfloat{%
        \includegraphics[width=0.290\textwidth]{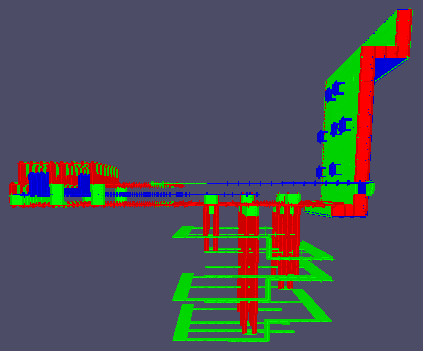}}
        \enskip
    \subfloat{%
        \includegraphics[width=0.425\textwidth]{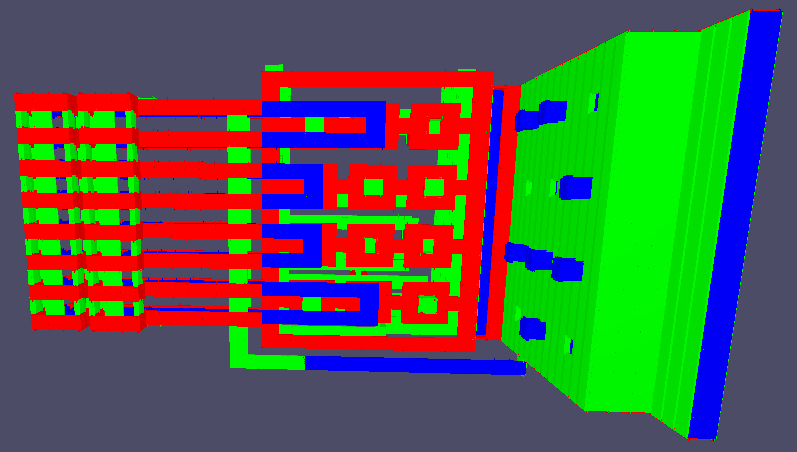}}
        \enskip
  \caption{Example assembly representing a 3SAT instance, visualized in FTAM simulator, used in the proof of Lemma~\ref{lem:rigidity-from-assembly-NP-hard}.}
  \label{fig:simulater_machine}
\end{figure}
\vspace{-10pt}

\vspace{-20pt}
\subsection{Determining the terminality of an assembly is co-NP-complete}\label{sec:assembly-terminality}

In addition to rigidity, terminality is another useful-to-know property of assemblies. Using much of the same logic from the previous result, we can prove a similar result regarding the terminality of arbitrary assemblies.%

\begin{problem}[Terminality-from-assembly]
Given an FTAM system $\mathcal{T}$ and assembly $\alpha \in \prodasm{\mathcal{T}}$, is $\alpha$ terminal?
\end{problem}

\begin{theorem}\label{thm:terminality-from-assembly}
Terminality-From-Assembly is co-NP-complete.
\end{theorem}
\vspace{-5pt}

To prove Theorem~\ref{thm:terminality-from-assembly}, we prove the following two lemmas.

\begin{lemma}\label{lem:terminality-from-assembly-co-NP}
  The complement of terminality-from-assembly is in NP.
\end{lemma}

\begin{proof}
For an instance of the problem, we are given an FTAM system $\mathcal{T} = (T,\sigma,\tau)$ and assembly $\alpha$.  Our certificate in this case includes a configuration $c_\alpha$ for the assembly $\alpha$ and a frontier location $f$. Similar to in the proof of Lemma~\ref{lem:rigidity-from-assembly-co-NP}, (and since the encoding of $f$ requires space $\le |\alpha|)$ we know that the certificate is polynomial in size to $\alpha$.  Also, we can check the validity of configuration $c_\alpha$ in polynomial time. Now, we simply need to verify the frontier location $f$ (a) isn't already occupied by a tile and (b) is adjacent to tiles in $\alpha$ while it's in configuration $c_\alpha$ such that the adjacent glues allow the tile specified by $f$ to bind to $\alpha$ with bonds collectively $\ge \tau$ strength, which can be done in time $O(|T|) + O(|\alpha|)$. Since the certificate has polynomial size in relation to $\alpha$ and can be verified in polynomial time to show that $\alpha$ is not terminal, the problem of determining if $\alpha$ is terminal is in co-NP.
\end{proof}

Now, we will also show that the complement of terminality is NP-hard. A slight augmentation to the 3SAT machine assembly can be used to achieve this.

\begin{lemma}\label{lem:terminality-from-assembly-NP-hard}
  The complement of terminality-from-assembly is NP-hard.
\end{lemma}
\vspace{-5pt}

To prove Lemma~\ref{lem:terminality-from-assembly-NP-hard}, we use almost identical techniques as for the proof of Lemma~\ref{lem:rigidity-from-assembly-NP-hard}, with a slight modification to the 3SAT machine so that, if and only if the 3SAT instance is satisfiable, then there will be a valid configuration of the assembly which represents the satisfying assignment, and in that configuration - and no other valid configuration - there will be a frontier location, which means that the assembly is not terminal.

Theorem~\ref{thm:terminality-from-assembly} is proven by Lemmas~\ref{lem:terminality-from-assembly-co-NP} and \ref{lem:terminality-from-assembly-NP-hard}.
\vspace{-10pt}

\bibliographystyle{splncs04}
\bibliography{FTAM,tam}

\newpage
\newgeometry{margin={1in,1in}} %

\appendix

\section{Technical Details for Section \ref{sec:definitions}}\label{sec:definitions-append}

\begin{algorithm}[H]
\caption{A procedure to perform one step of the self-assembly process of FTAM system $\mathcal{T}$}
\label{alg:assembly-step}
\begin{algorithmic}[1]
\Procedure{\texttt{ASSEMBLY-STEP}}{$\alpha$, $\mathcal{T}$} \Comment{Takes an assembly $\alpha$ and FTAM system $\mathcal{T}$}
    \If {$|\partial^{\mathcal{T}} \alpha| = 0$}
        \State \textbf{return} $\alpha$ \Comment{No frontier locations remain, $\alpha$ is terminal}
    \Else
        \State Uniformly at random select $(t, B) \in \hat{\partial}^{\mathcal{T}} \alpha$ \Comment{Select a frontier location}
        \State Attach a tile of type $t$ with bonds to tiles in $B$,  $\alpha \to_1^{\mathcal{T}} \alpha'$ \Comment{Add a tile}
        \State Uniformly at random select $c_\alpha' \in \mathcal{C}_{max}(\alpha')$ \Comment{Find new-bond-maximizing configuration}
        \State Form all bonds possible in $c_\alpha'$ to yield $\alpha''$ \Comment{Form those bonds}
        \State \textbf{return} $\alpha''$ \Comment{Return the new assembly}
    \EndIf
\EndProcedure
\end{algorithmic}
\end{algorithm}

\begin{algorithm}[H]
\caption{A procedure to perform the self-assembly process of FTAM system $\mathcal{T}$}
\label{alg:full-assembly}
\begin{algorithmic}[1]
\Procedure{\texttt{FULL-ASSEMBLY}}{$\alpha$, $\mathcal{T}$} \Comment{Takes an assembly $\alpha$ and FTAM system $\mathcal{T}$}
    \State $\alpha' = \texttt{ASSEMBLY-STEP}(\alpha, \mathcal{T}$)
    \If {$\alpha == \alpha'$}
        \State \textbf{return} $\alpha'$
    \Else
        \State \textbf{return} $\texttt{FULL-ASSEMBLY}(\alpha', \mathcal{T})$
    \EndIf
\EndProcedure
\end{algorithmic}
\end{algorithm}

\section{Technical Details for Section \ref{sec:robustness}} \label{sec:robustness_appendix}

\subsection{Filler tiles}

First, we discuss the ``filler tiles'' that are used to fill in edge frames. On perfectly square faces, this can trivially be done, with filler tiles allowing to attach as the assembly grows. However, in cases where the face has a concave corner, a rectangular decomposition of the face with each rectangle being assigned a unique filler tile would prevent the filler tiles from overgrowing their bounds.

\subsection{Vertices and perspectives} \label{sec:verts_and_perts}

We again provide the enumeration of all $2\times2\times2$ polycubes for reference.

\begin{figure}
\centering
    \includegraphics[width=\textwidth]{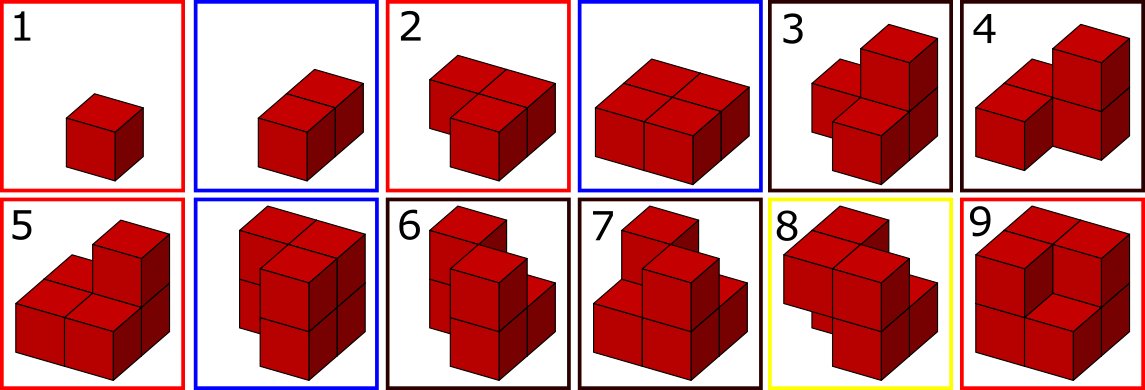}
    \caption{All possible polycubes that can fit inside of a $2\times2\times2$ space, and furthermore, all possible vertex types that could exist on a polycube.}
    \label{fig:possible_outline_vertices_appendix}
\end{figure}

We categorize the types of vertices into two groups, \emph{simple} and \emph{complex}. In the enumeration in Figure \ref{fig:possible_outline_vertices}, the polycubes in the blue squares actually don't have a vertex in the center. The vertices in the polycubes in red (1,2,5,9) have three edges and three faces incident on the center point, creating what we refer to as a \emph{simple} vertex. Of these 4 vertices, 1 and 9 are the same vertex type, which we will refer to as a \emph{convex} vertex, and 2 and 5 are the same vertex type, which we will refer to as a \emph{concave} vertex. The vertices in the polycubes in brown (3,4,6,7) have more than three edges and more than three faces incident on the center point, creating what we refer to as a \emph{complex} vertex. All of these complex vertices are unique, and we will refer to them by their number. The polycube in yellow (8) is a special case in which there are more than three edges and more than three faces incident on the center point, but the polycube is arranged in a way that the center point can be thought of as two different simple convex vertices, one for each location that is missing a cube.

A vertex can be \emph{symmetric}, meaning all edges have the same perspective, \emph{semi-symmetric}, meaning some edges have the same perspective, or \emph{asymmetric}, meaning no edges have the same perspective. Semi-symmetric and asymmetric vertices have to differentiate between the different perspectives the vertex can exist in. The simple convex vertex and vertex 3 are both symmetric, meaning they have only one perspective each. Vertex 4 and 7 are semi-symmetric, with vertex 4 having 2 perspectives (even though it has 4 edges) and vertex 7 having 3 perspectives (even though it has 6 edges). The simple concave vertex and vertex 6 and asymmetric, with the simple concave vertex having 3 perspectives (and 3 edges) and the vertex 6 having 5 perspectives (and 5 edges). An example of different perspectives can be seen in the difference between Figure \ref{fig:concave_vertex} and Figure \ref{fig:combined_vertex}. All together, there are 15 different perspective, meaning we need 15 different tiling protocols to handle all situations.

We will now look at the loop lengths and bond sequences of different vertices. For bond sequences, we will use the notation $(b,b, ... ,b)$ where $b \in \{R,F\}$ and $R$ stands for rigid and $F$ stands for flexible. The first bond in the sequence will represent the edge assembling up to the vertex in the assembly process, and will therefore always be flexible. The other bonds in the sequence will represent the other bonds within the loop of tiles incident on the vertex, following the ordering set by the loop in either direction, without loss of generality. The set of all bond sequences for a single vertex is a non-repeating cyclic permutation group, minus the elements that begin with a rigid bond instead of a flexible bond.

We start with the simple vertices. First is the simple convex vertex. It has a loop length of 3 tiles and a bond sequence $(F,F,F)$. It is symmetric, meaning we don't have to differentiate between the edges. Therefore, to make a simple convex vertex, an edge will just initiate a protocol where the last two tiles of the edge end in flexible glues and another tile that matches both glues attaches to complete the loop. The next vertex is the simple concave vertex. It has a loop length of 5 and can have a bond sequence of either $(F,R,R,F,F)$, $(F,F,R,R,F)$, or $(F,F,F,R,R)$. Using each bond orientation will result in a different perspective, meaning all 3 perspectives can be deterministically assembled. Therefore, attaching the loop of tiles with the first bond sequence will yield the vertex in \ref{fig:combined_vertex}, the second bond sequence will yield the vertex in \ref{fig:concave_vertex}, and the third sequence will yield the mirror opposite of the vertex in \ref{fig:combined_vertex}.

The most unique complex vertex is vertex 6. It has a loop length of 7 and can have a bond sequence of either $(F,R,F,R,F,F,F)$, $(F,F,R,F,R,F,F)$, $(F,F,F,R,F,R,F)$, $(F,F,F,F,R,F,R)$, or $(F,R,F,F,F,R,F)$. As with the simple concave vertex, each bond sequences results in a different perspective, allowing the system to differentiate between the 5 different perspectives of vertex 6. Similarly, vertex 4 also has its own unique combinations. It has a loop length of 6 and can have a bond sequence of $(F,R,F,F,R,F)$ or $(F,F,R,F,F,R)$. The two bond orientations correspond to the two perspectives of vertex 4, allowing both to be deterministically assembled.

Vertex 3 and 7, however, share a combination of a loop length and bond sequence. Both have a loop length of 6 and a bond sequence of $(F,F,F,F,F,F)$. Because of this, attaching a loop of 6 tiles using all flexible bonds at the end of an edge can result in either vertex. In addition, since flexible bonds can ``mimic'' rigid bonds using a ``Straight'' orientation, the loop of tiles can even configure into vertex 4. Note that the reverse is mitigated by the fact vertex 4 has rigid bonds and no bonds in vertex 3 or 7 are ``Straight''. All together, it can end up in the one perspective from vertex 3, one of the two perspectives from vertex 4, or one of the three perspective from vertex 7. Given all these possibilities, vertex 3 and 7 cannot be deterministically assembled.

\begin{figure}[htp]
\centering
    \includegraphics[width=\textwidth]{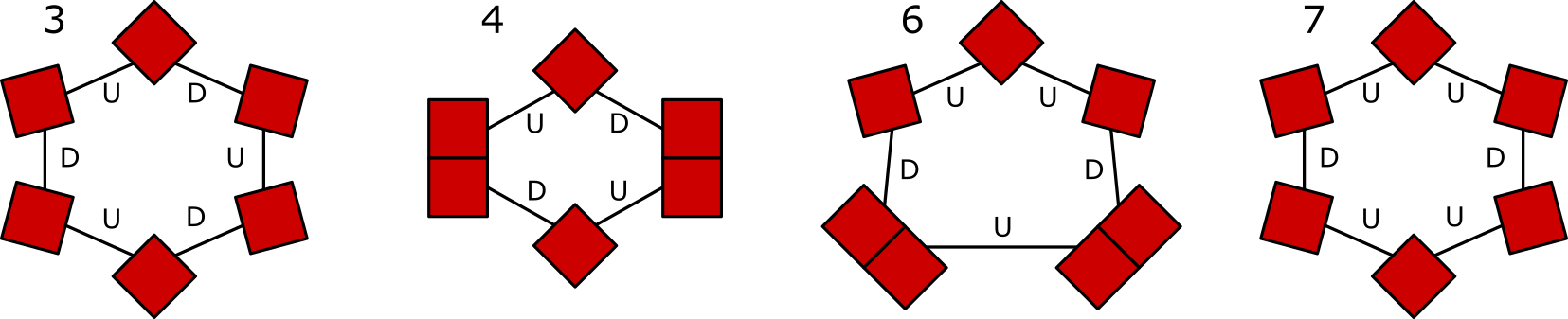}
    
    \caption{The loops of tiles that form each complex vertex.}
    \label{fig:complex_vetex_loops}
\end{figure}

\subsection{Edge growth protocol} \label{sec:edge_protocol}

When attempting to deterministically assemble outlines of polycubes in the FTAM, edges in the edge frame can be grown using the following protocol. Side 1 and side 2 refer to the two columns of tiles on the two faces that make up the edge. The hinge refers to the series of flexible bonds between tiles on side~1 and tiles on side~2. \begin{enumerate}
    \item An exposed rigid double strength glue on side 1 of the edge will attach a new tile $t$ on side 1,
    \item A flexible glue on tile $t$ on side 1 and an exposed rigid glue on side 2 will cooperate to attach a new tile $t'$ on side 2 of the edge,
    \item A rigid double strength glue on tile $t'$ on side 2 of the edge will attach a new tile $t''$ on side 2, and
    \item A flexible glue on tile $t''$ on side 2 and a rigid glue on tile $t$ on side 1 will cooperate to attach a new tile $t'''$ on side 1 of the edge
\end{enumerate} An edge can grow indefinitely by repeating this process using unique glues to grow up to a certain length. Notice that each new tile attaches using at least one rigid bond, meaning that, additional flexibility cannot be added to the edge past the flexibility of the hinge. Furthermore, there will only ever be one frontier location (per configuration, if multiple, but these are the same tiles and bonds) on the edge at any assembly step, leaving no room for non-determinism.

\subsection{Combating multiple edge frames} \label{sec:combating_multiple_edge_frames}

The two methods of combating multiple edge frames mentioned in the main section are blocking and symmetry. Blocking refers to when the faces surrounded by one edge frame would collide with the faces of another if the chiralities of the edge frames disagreed. This is actually the case in the example given in Figure \ref{fig:dumbbell_shape}. In these situations, even if the additional edge frames are configured to the wrong chirality during the assembly process, eventually the tiles with the potential to collide will be added to the assembly and force the correct chirality of both edge frames with respect to each other. In Figure \ref{fig:dumbbell_path}, you can see on the right how the inversion of the additional edge frame would cause the collision of tiles in the assembly. In this example, the yellow tile on the left of Figure \ref{fig:dumbbell_path} would force the correct relative chirality.

One other aspect of a shape with multiple edge frames that may reduce the number of possible configurations it can exist in is, like with full shapes, symmetry of the edge frames. To utilize the same example, imagine if the shape in Figure \ref{fig:dumbbell_shape} did not contain the yellow tile on left end or the symmetric tile on the right end. In this case, the ends would be free to reconfigure into the wrong chirality. However, this doesn't result in 4 different shapes that the mismatching chirality of 2 additional edge frames should produce. Instead, inversion of the left end of the assembly and not the right yields the same shape as inversion of the right end of the assembly and not the left, resulting in only 3 different possible shapes.

\begin{figure}[ht]
\centering
    \subfloat[][The full dumbbell shape, a narrow connection piece and two end pieces]{%
        \label{fig:dumbbell_shape}%
        \includegraphics[width=0.7\textwidth]{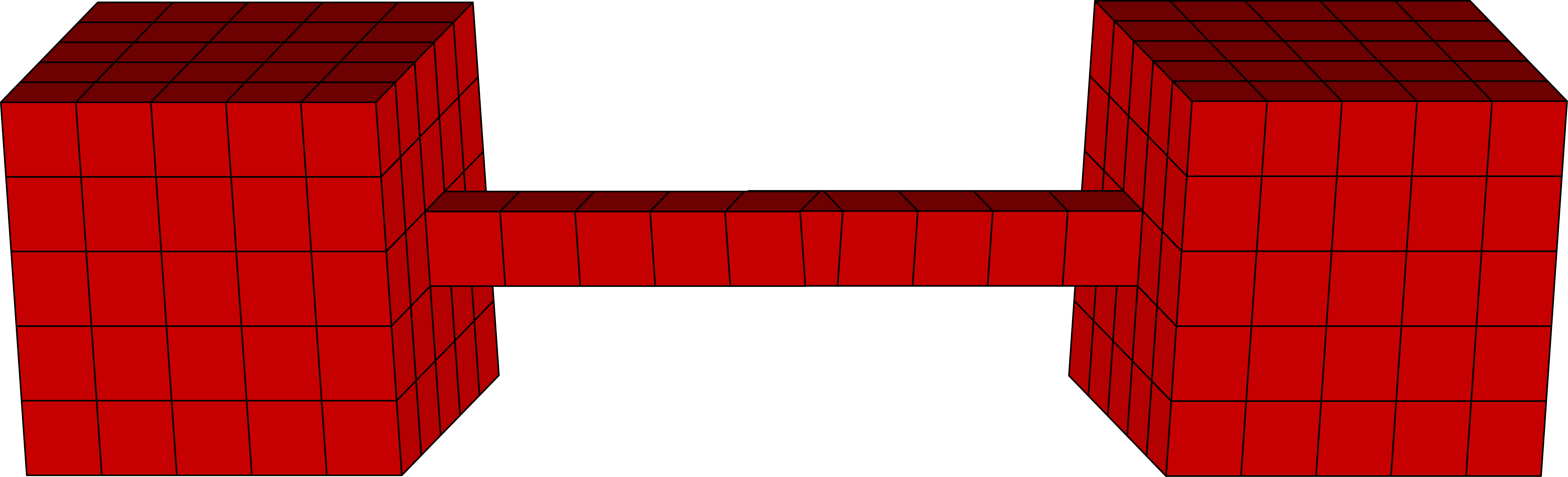}
        }%
        \quad
  \subfloat[][An outline of the edges in the shape, showing 3 different edge frames]{%
        \label{fig:dumbbell_edges}%
        \includegraphics[width=0.7\textwidth]{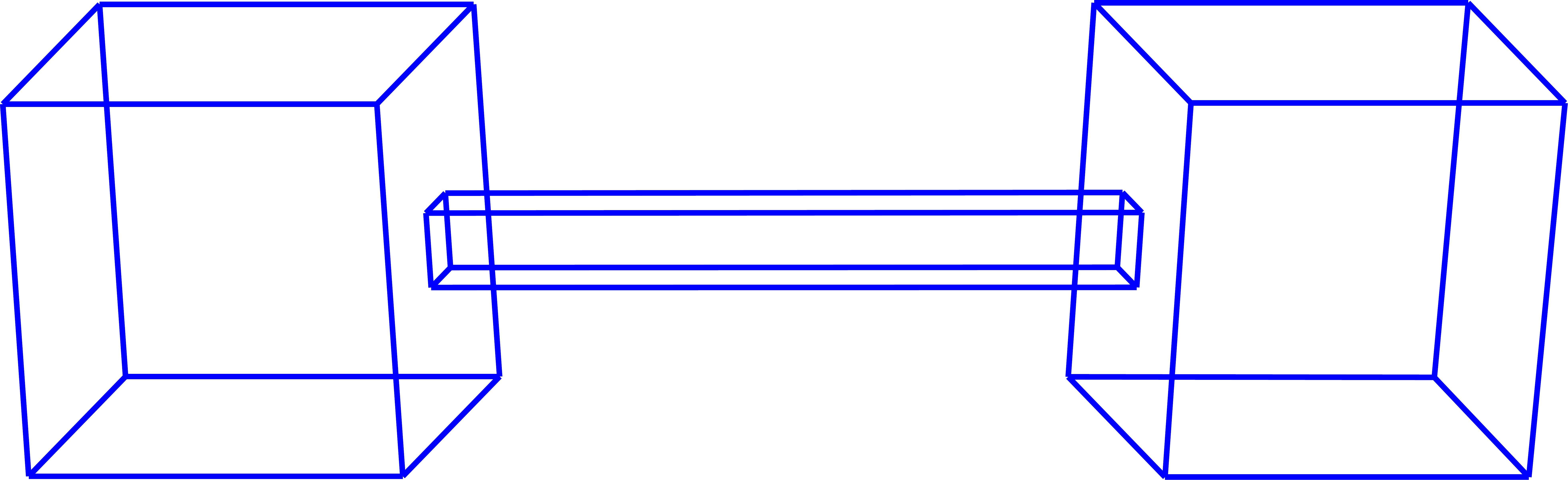}
        }%
        \quad
  \subfloat[][The rights shows how an inversion of an end piece would cause a collision. The left shows the tile that needs to be placed to force the correct orientation of the end pieces (the yellow tile).]{%
        \label{fig:dumbbell_path}%
        \includegraphics[width=0.7\textwidth]{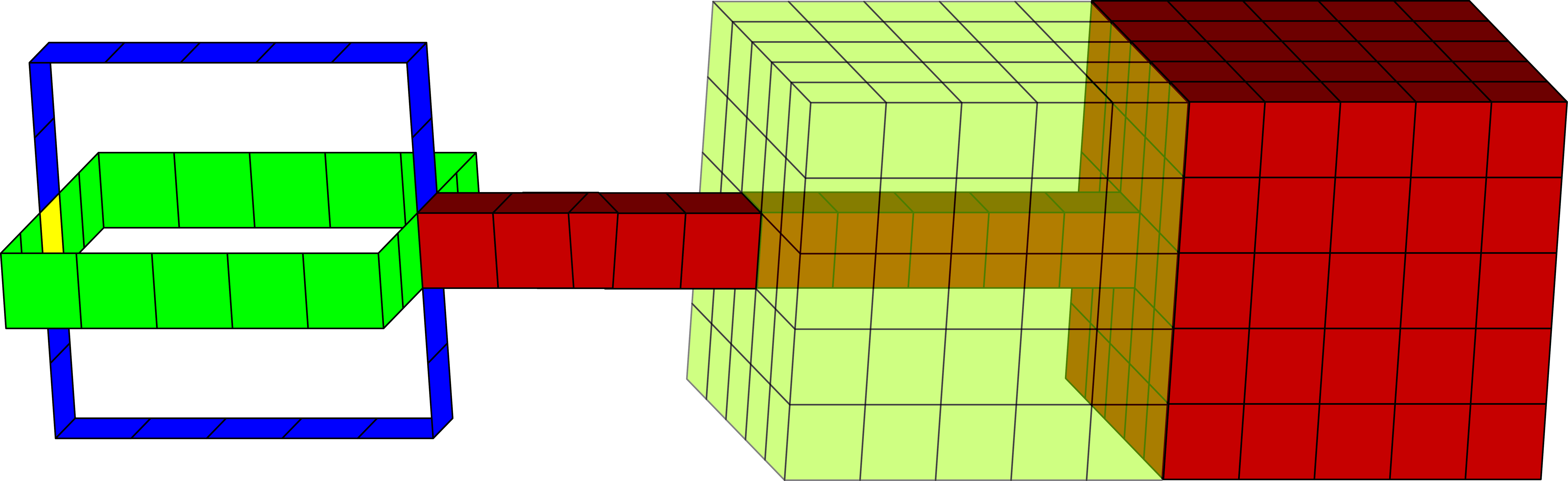}
        }%
        \quad
  \caption{Example of a shape with multiple edge frames}
  \label{fig:dumbbell}
\end{figure}

\section{Technical Details for Section \ref{sec:utilizing_flexibility}} \label{sec:utilizing_flexibility_appendix}

\subsection{Staged Functional Surface: Maximizing the number of reconfigurations}

For this section, we simply provide another visual to illustrate how secondary stages could be designed to make certain functional surfaces. This visual is provided in Figure \ref{fig:staged_functional_surface_example}, while Figure \ref{fig:staged_functional_surface_appendix} was included again for the reader's reference.

\begin{figure}[htp]
\centering
    \includegraphics[width=0.9\textwidth]{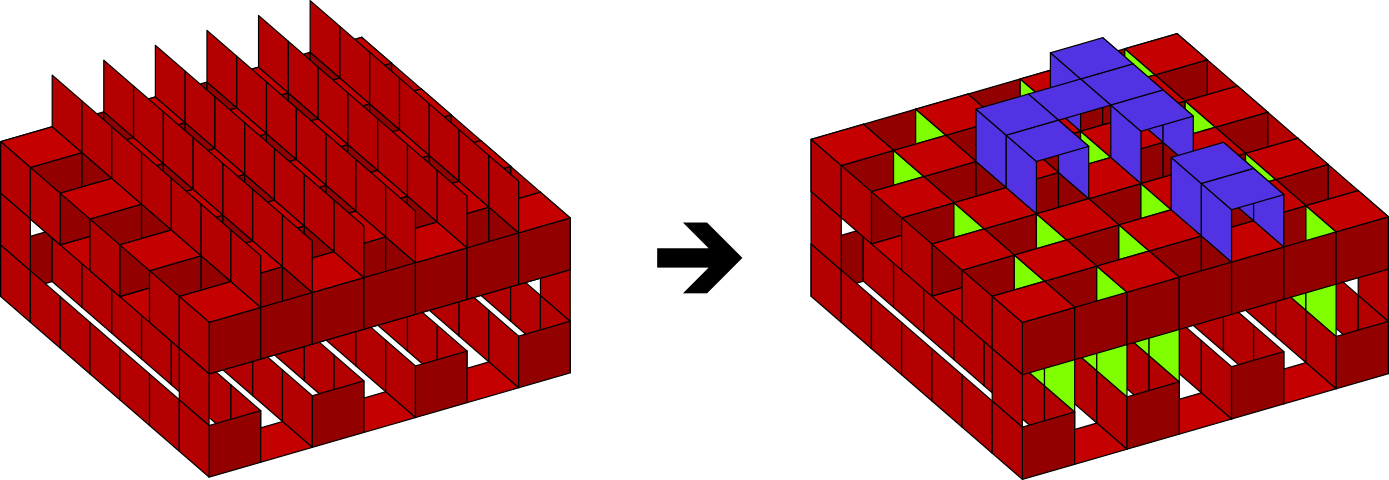}
    \caption{An example of reconfigurable shape that can be used in a staged environment to display a functional surface only after additional tiles are added to the system.}
    \label{fig:staged_functional_surface_appendix}
\end{figure}

\begin{figure}[htp]
\centering
    \includegraphics[width=\textwidth]{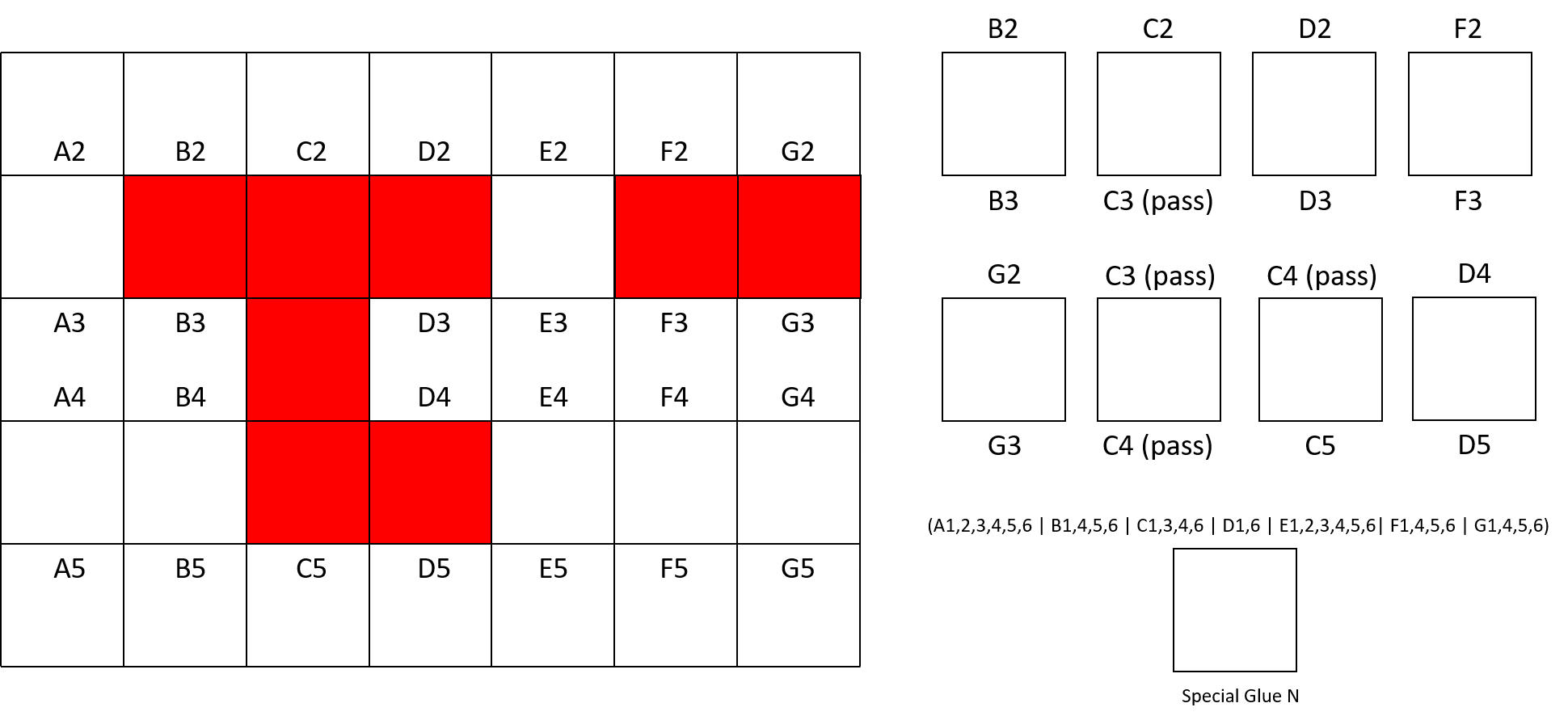}
    \caption{An example of a system that could be used to match pixels with tiles. Here, the tiles are shown that are used to make the image on the right of Figure \ref{fig:staged_functional_surface_appendix}. The tiles in the bottom row on the right can be added to pin the unused pixel location down into the surface of the assembly.}
    \label{fig:staged_functional_surface_example}
\end{figure}

\subsection{Compressing/expanding structures}\label{sec:compress-expand=append}

In this section we provide technical details about the construction in Section~\ref{sec:compress-expand}.

Let $\mathcal{T}$ be a staged FTAM system which is capable of self-assembling assembly $\alpha$ which can be thought of as a set of $6$ ``sheets'' that are each approximately $n \times n$ squares of tiles (which can be seen laid out flat in Figure~\ref{fig:compression-animation}f and in more detail in Figure~\ref{fig:unfolded-labels} where $\alpha$ contains each $n \times n$ square minus the outer perimeter strips which are colored grey). The ordering of growth of the squares follows the path of the arrow shown in Figure~\ref{fig:unfolded-labels}, and the squares are built using standard aTAM techniques to efficiently build $n \times n$ squares using $O(\log(n))$ tile types \cite{RotWin00}.

For the second stage, if the expanded cube is desired, tile types are added which attach in counterclockwise order along the perimeter to form one additional row of tiles and fill out the squares to size $n \times n$ (starting from the topmost left corner as shown in Figure~\ref{fig:unfolded-labels}) so that the glues of the outside perimeter on the sides labeled in Figure~\ref{fig:unfolded-labels} are flexible strength-1 glues with the labels shown.  These perimeter tiles also have rigid glues between each other (except at the boundaries between squares) which will provide the rigid frame of the cube (while the other glues between tiles in the interior of the squares are flexible).  A constant number of unique tile types are sufficient to tile these perimeter locations in this second stage and thus form the rigid cube (shown in Figure~\ref{fig:compression-animation}a).

If, instead, the compressed structure is desired, for the second stage the tiles added to form the perimeter create an alternating pattern of glues and grow $2 \times 2$ \emph{tabs} as shown in Figure~\ref{fig:wrinkle-labels}, which forces them to fold into the wrinkled pattern shown in Figure~\ref{fig:compressed}.  The tabs form ``caps'' which stably bind the structure. Note that for clarity, in Figure~\ref{fig:wrinkle-labels}, the labeled locations are given unique labels to make it easier to see which locations will bind, but due to the dynamics of the FTAM, as a square self-assembles, the perimeter glues will force it to fold into the desired pattern even if a constant-sized set of glue types is repeatedly used. This is because the ordering of perimeter growth is fixed (by cooperative growth) and the fact that the tabs are $2 \times 2$ squares allows enforcement of the growth ordering so that each tab is completed before perimeter growth proceeds beyond it. Since the squares fold into approximately $\sqrt{n}$ length wrinkled layers, the constant-sized pattern must repeat each $O(\sqrt{n})$ distance, which requires $O(\log(\sqrt{n}))$ unique tile types which can be embedded as part of the counter tile types which form the original squares (and thus the tile complexity remains at $O(\log(n))$, causing marker glues to be positioned with correct periodicity in case the tile types for forming the compressed structure are added.  Otherwise, the tile types which force formation of the expanded cube still attach to those glues, but they don't propagate this pattern to the perimeter.  Therefore, only a constant number of unique tile types are required for the second stage to form this configuration as well.

\begin{figure}[htp]
\centering
    \includegraphics[width=3.0in]{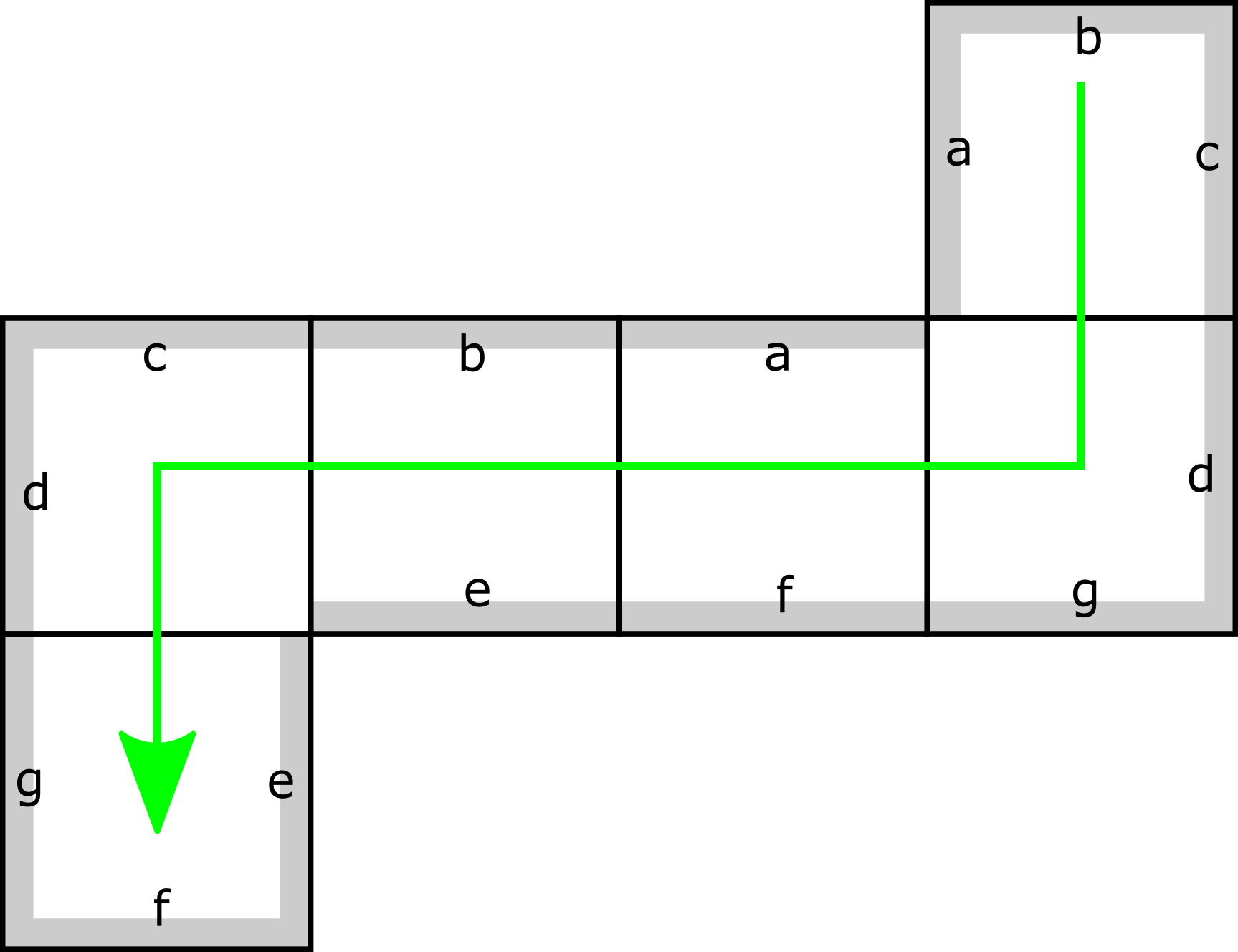}
    \caption{Label types for glues on each side.}
    \label{fig:unfolded-labels}
\end{figure}

\begin{figure}[htp]
\centering
    \includegraphics[width=\textwidth]{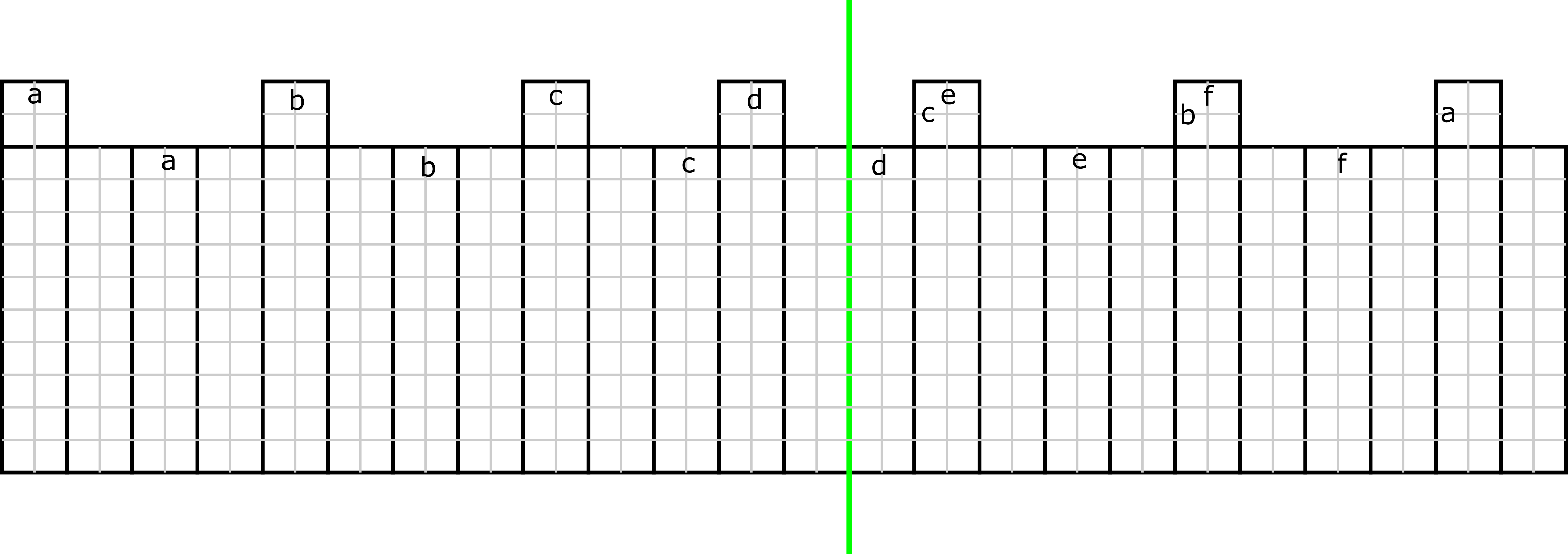}
    \caption{Schematic depiction of a portion of the perimeter of a square where additional tiles attach, forming tabs, and glue labels are exposed which are needed to make a portion of a sheet fold into two ``wrinkled'' layers. The green line shows the boundary between two separate wrinkled layers.  The left side will fold to make a wrinkled layered from left-to-right, and the right side will fold underneath to form a second wrinkled layer from right-to-left, also binding to the layer above it to form a stable structure.}
    \label{fig:wrinkle-labels}
\end{figure}

\begin{figure}[H]
\hfill
\includegraphics[width=1.0in]{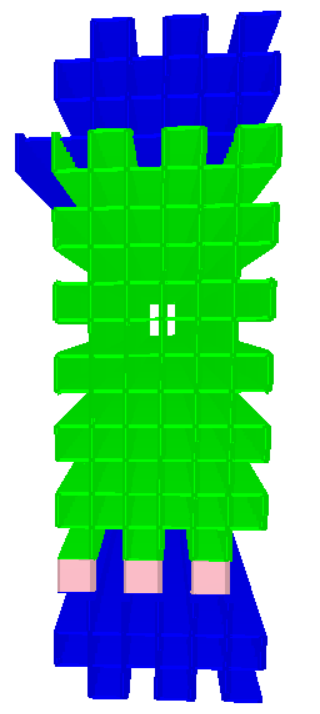}
\hfill
\includegraphics[width=1.0in]{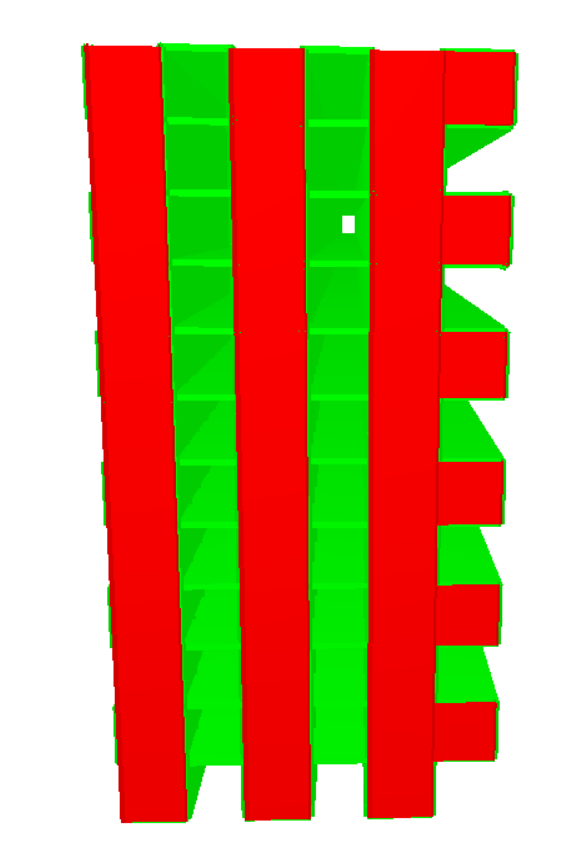}
\hfill
\,
\caption{(Left) The compressed configuration of the tiles which can also form a cube (shown without end~``caps''). (Right) A portion of the compressed configuration with end ``caps'', which fold from the tabs, included.}\label{fig:compressed}
\end{figure}

\section{Technical Details for Section \ref{sec:complexity}} \label{sec:complexity_appendix}

\subsection{Rigidity-from-system is uncomputable: Technical details}

Here we provide an illustration in Figure~\ref{fig:turing_machine} of an embedding of our construction at three different points.

\begin{figure}[htp]
\centering
  \subfloat{%
        \label{fig:flexible_turing_machine}%
        \includegraphics[width=0.4\textwidth]{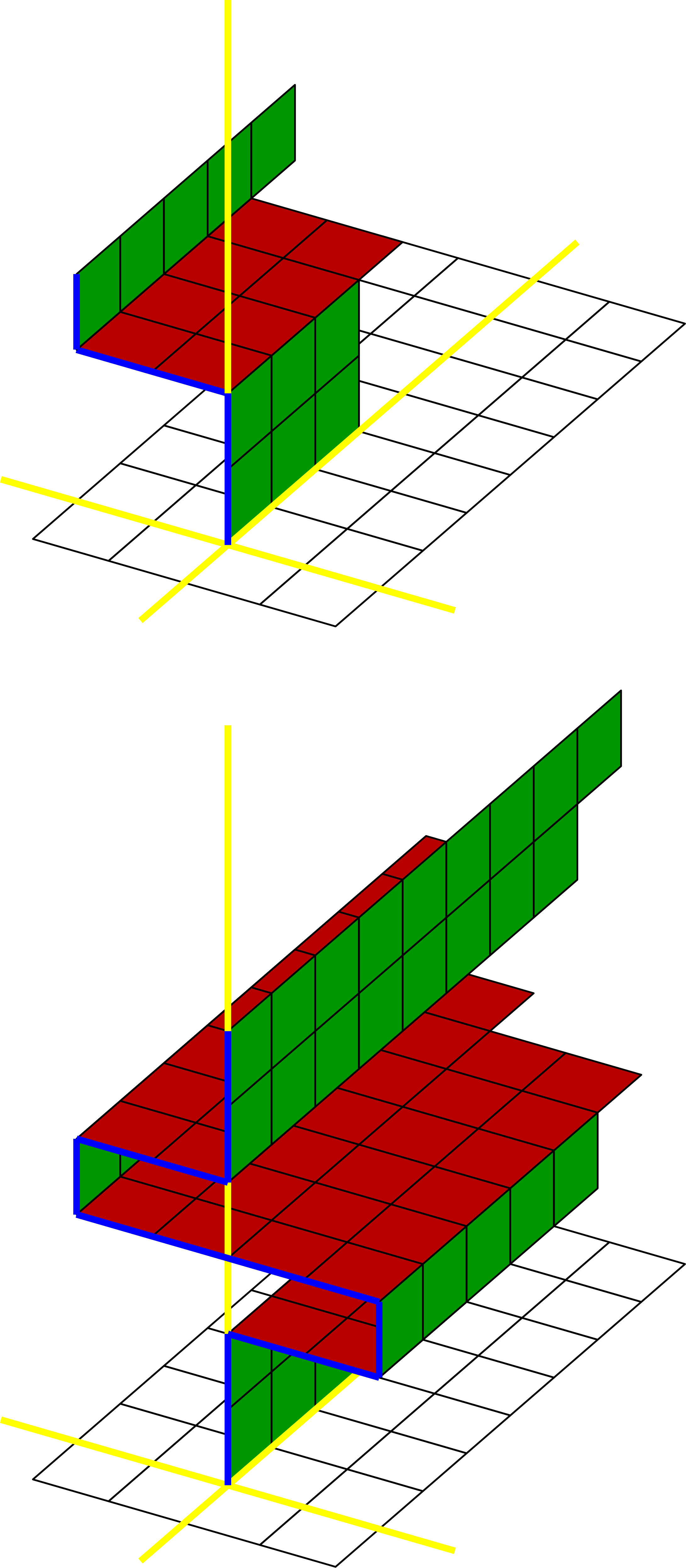}
        }%
  \subfloat{%
        \label{fig:rigid_turing_machine}%
        \includegraphics[width=0.4\textwidth]{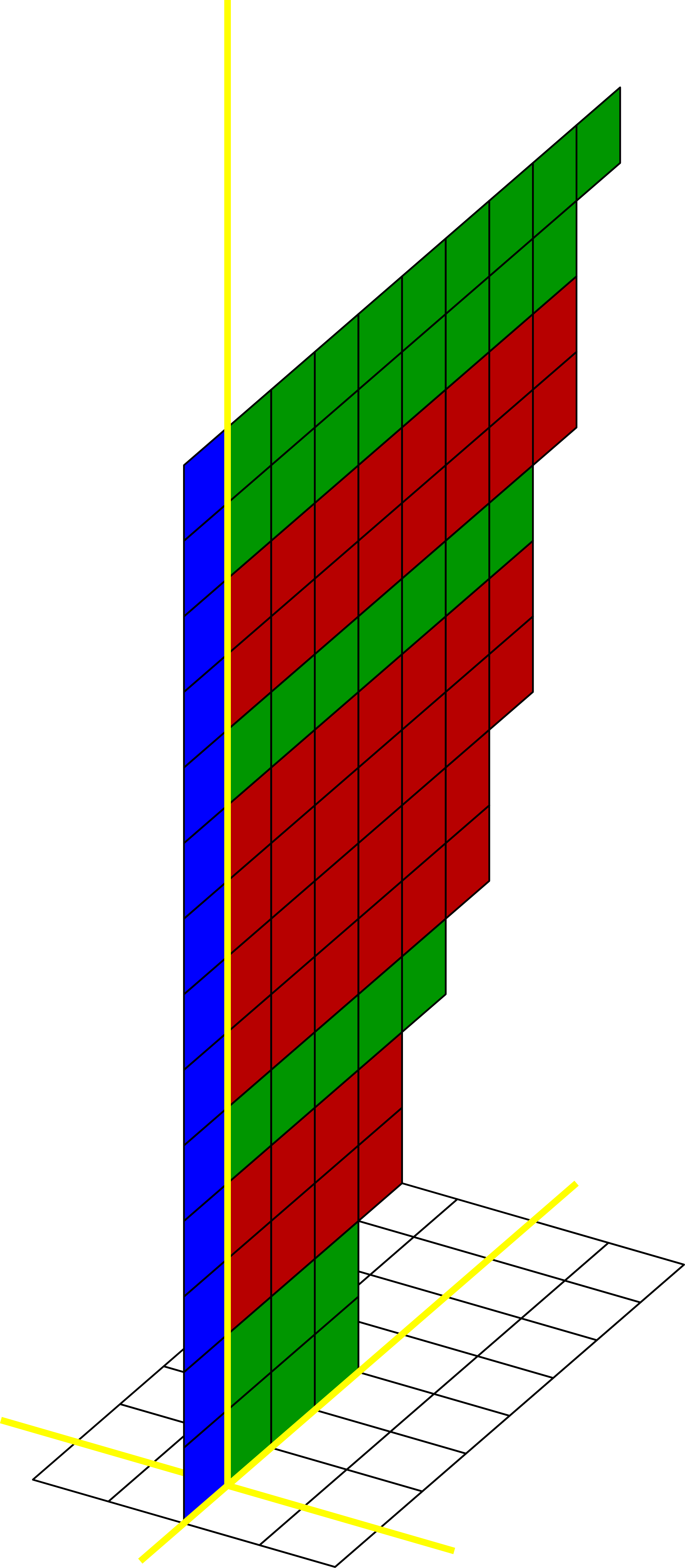}
        }%
        \quad
  \caption{(a) As the Turing machine is being simulated, the assembly is grown with flexible bonds, allowing every row to be flexible relative to one another and the entire assembly to fold back and forth like an accordion. (b) If the Turing machine ever halts, a ``backbone'' of tiles connected through rigid bonds is grown on the side of the assembly, making the entire structure rigid.}
  \label{fig:turing_machine}
\end{figure}

First, we consider a general structure of a commonly used type of aTAM  tile  assembly  system  for  simulating  the  behavior  of  Turing  machines. A zig-zag aTAM system is one which grows in a strict row-by-row ordering.  More specifically, the first row grows either left-to-right or right-to-left, completely, at which point the second row begins growth in the opposite direction.  When it completes, the third grows, again in reversed direction, and so on. Given a Turing machine $M$, $M$ can be simulated by a temperature-$2$ zig-zag aTAM system, $\mathcal{P}$ say, such that if $M$ halts,  a final ``halting'' tile attaches to the westernmost column in the northernmost row of the terminal assembly of $\mathcal{P}$. One can also show that $M$ can be simulated by a zig-zag system which produces a single terminal assembly such that the westernmost tiles of this assembly (possibly including the halting tile) are colinear.   Moreover, such an aTAM system gives rise to an FTAM system, $\calT$ say, where the tile types of $\calT$ are identical to the tile types of the aTAM system and all glues are rigid. 
To show Theorem~\ref{thm:Rigidity}, we consider the FTAM system $\calT'$ that is obtained from $\calT$ by (1) modifying the glues on north and south edges so that they are flexible, and (2) adding appropriate glues to tile types and tiles to the tile set of $\calT$ so that if $M$ halts,  tiles of these types initially bind to the west edge of the halting tile (via an added strength-$2$ glue) of the terminal assembly of $\calT'$, and then cooperatively bind one at a time along the west edges of the westernmost tiles of this terminal assembly (via strength-$1$ glues added to these west edges and north/south glues of the tiles of the additional tile types) to form a single tile wide column of tiles each of which is bound to a westernmost tile in the terminal assembly of $\calT'$.  Moreover, the north and south glues of the tile types that bind to form the column of westernmost tiles are rigid. We call such a column of tiles a ``backbone''. Then, as the east/west glues of tiles belonging to any row of tiles in the terminal assembly of $\calT'$ are rigid, and since a backbone of tiles self-assembles iff $M$ halts, we see that the terminal assembly of $\calT'$ is rigid iff $M$ halts. In other words, we have a system such that for any terminal assembly $\alpha$, $\alpha$ is rigid iff $M$ halts. See Figure~\ref{fig:turing_machine} for an intuitive description of self-assembly in $\calT'$ for a simulation of a machine that halts. This suffices to show Theorem~\ref{thm:Rigidity}.

\subsection{Rigidity-from-assembly is co-NP-complete: Technical details} \label{sec:assembly-rigidity-append}

Here we provide technical details for the proof of Lemma~\ref{lem:rigidity-from-assembly-NP-hard}.

To prove Lemma~\ref{lem:rigidity-from-assembly-NP-hard}, we will reduce 3SAT to the complement of the rigidity-from-assembly problem. For our reduction, we introduce a new construction that we will subsequently refer to as the \emph{3SAT machine}. This is a computable assembly that is made up of four gadgets that connect together and will be able to reconfigure if and only if the corresponding 3SAT formula has a satisfying assignment.

\begin{figure}[htp]
\centering
    \includegraphics[width=0.70\textwidth]{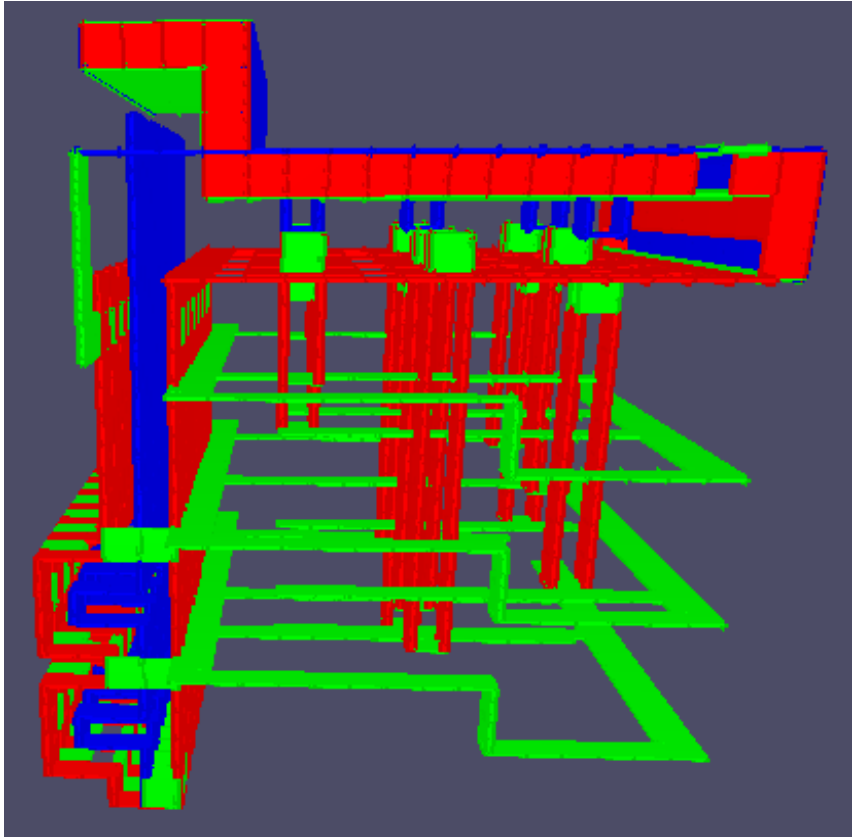}
    \caption{The trivial state that the assembly can exist in, as shown in the simulator}
    \label{fig:trivial_in_simulator}
\end{figure}

\begin{figure}[htp]
\centering
    \includegraphics[width=0.80\textwidth]{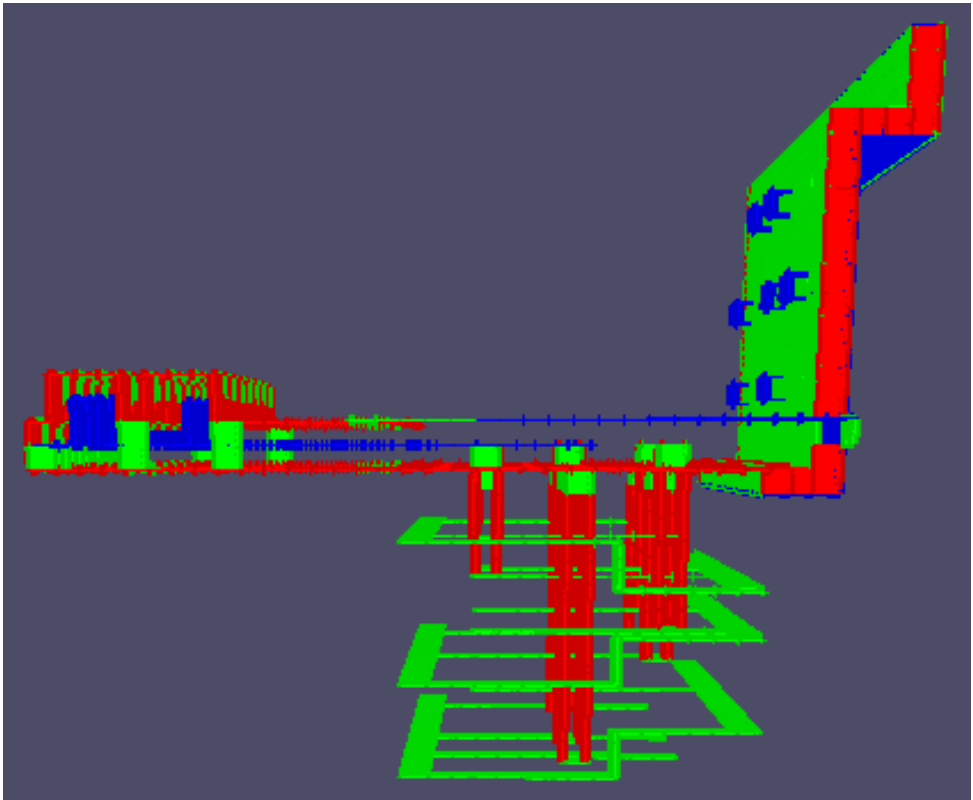}
    \caption{The satisfied state, as shown in the simulator from the $-Y$ perspective}
    \label{fig:satisfied_in_simulator_side}
\end{figure}

\begin{figure}[htp]
\centering
    \includegraphics[width=1.00\textwidth]{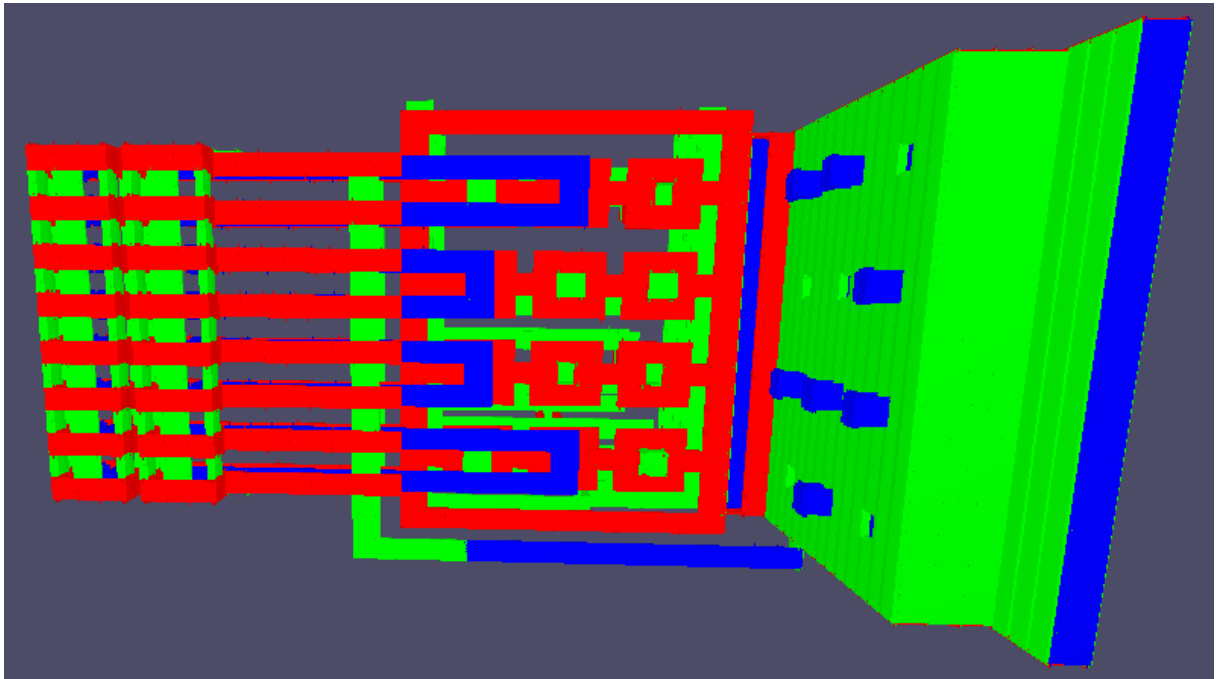}
    \caption{The satisfied state, as shown in the simulator from the $+Z$ perspective}
    \label{fig:satisfied_in_simulator_top}
\end{figure}

\begin{figure}[htp]
\centering
    \subfloat{%
        \label{fig:trivial_state}%
        \includegraphics[width=0.2\textwidth]{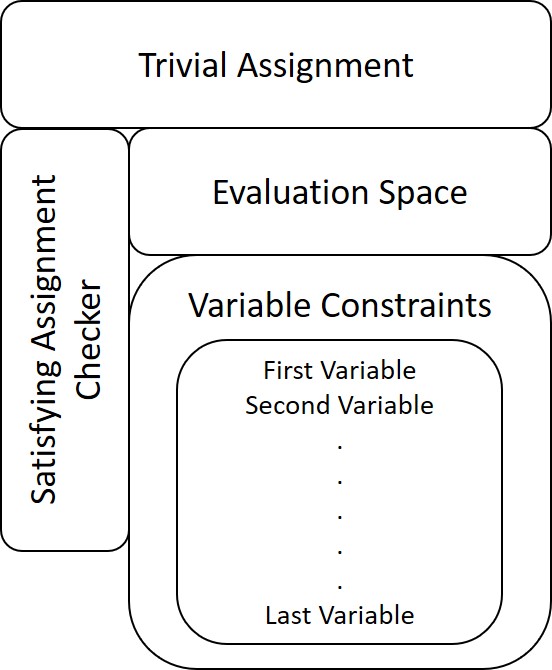}
        }
    \enskip
    \subfloat{%
        \label{fig:satisfying_state}%
        \includegraphics[width=0.35\textwidth]{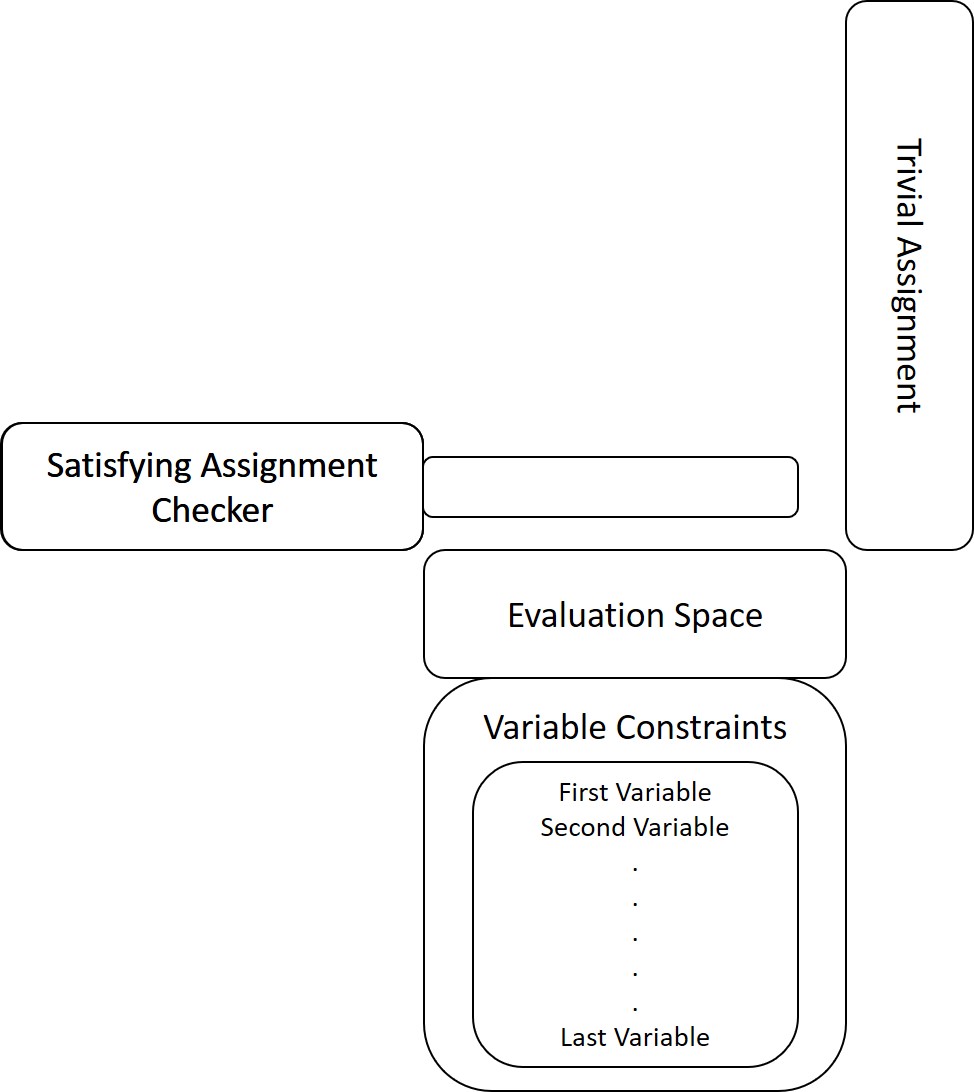}
        }
    \subfloat{%
        \label{fig:state_compass}%
        \includegraphics[width=0.18\textwidth]{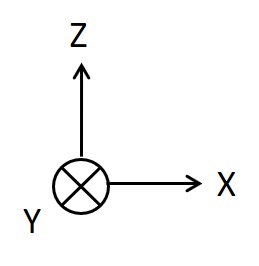}
        }
    \enskip
    \subfloat{%
        \label{fig:invalid_state}%
        \includegraphics[width=0.22\textwidth]{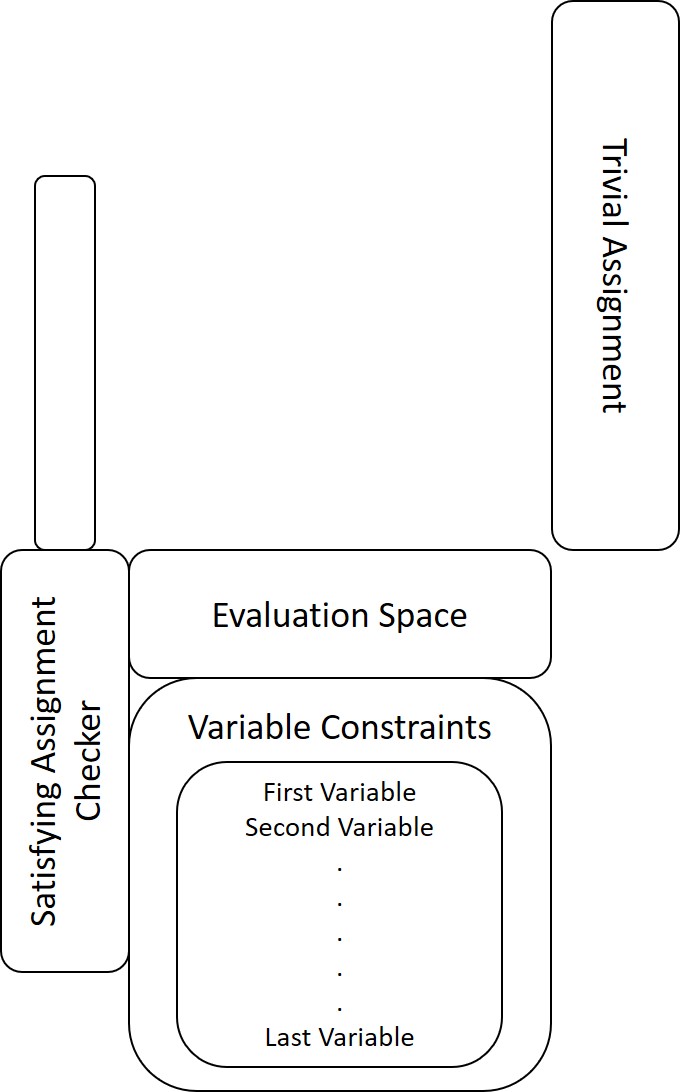}
        }
    \enskip
  \caption{Mock-up of (a) the initial state, (b) the alternate state, (c) the axes, and (d) the invalid state.}
  \label{fig:machine}
\end{figure}

Before we get into the details, we define terms and prove lemmas that we will use as tools in the overall proof. We define \emph{entanglement} to be the cause-effect relationship of connections within one set of faces in which the inversion of one connection causes the inversion of all other connections. In other words, a set of faces that have been entangled form a \emph{rigid component}, i.e. a subassembly that cannot reconfigure without fully inverting into a chiral configuration.

We define a \emph{traditional 4-sided loop}, as a cycle of faces in which each face is bound on opposite ends to two other faces. In other words, both pairs of faces that are opposite of each other will exist in planes in the same orientation $\{XY, YZ, ZX\}$. Notice that a traditional 4-sided loop will always be a rigid component. The four flexible bonds between the four faces that make up the loop will always have to be either a $\{U,U,U,U\}$ configuration or a $\{D,D,D,D\}$ configuration.

\begin{claim}\label{clm:one_piece_not_inverted}
    Given an assembly $\alpha$ and a configuration $c_\alpha$, for any rigid component $comp_{rigid}$ in $\alpha$, if $\alpha$ has another valid, non-trivial configuration $c_\alpha'$ that it can exist in, then there must also exist another valid, non-trivial configuration $c_\alpha''$, where $c_\alpha'$ may or may not equal $c_\alpha''$, in which the rigid component $comp_{rigid}$ is not inverted.
\end{claim}

\begin{proof}
    Assume that in $c_\alpha'$, $comp_{rigid}$ is not inverted. Then $c_\alpha' = c_\alpha''$. However, if we assume that in $c_\alpha'$, $comp_{rigid}$ is inverted, then $c_\alpha' = inverse(c_\alpha'')$.
\end{proof}

\begin{claim}\label{clm:one_piece_rigid}
    As a corollary of Claim~\ref{clm:one_piece_not_inverted}, for any rigid component $comp_{rigid}$ in $\alpha$, if $\alpha$ in $c_\alpha$ does not have another valid configuration with $comp_{rigid}$ in the same orientation, $\alpha$ itself is a rigid component.
\end{claim}

These claims give us a powerful tool in proving that pieces of an assembly are rigid components, since it allows us to assume one smaller rigid component will not reorient and then examine if any other pieces can reorient in relation to it.

\begin{claim}\label{clm:two_piece_entanglement}
    If two rigid components $comp_{rigid}^1$ and $comp_{rigid}^2$ share two faces $p_1$ and $p_2$ that are not coplanar, then $comp_{rigid}^1$ and $comp_{rigid}^2$ can be entangled into one rigid component.
\end{claim}

\begin{proof}
    Proof by contradiction. Assume $comp_{rigid}^1$ can reorient relative to $comp_{rigid}^2$. Since $comp_{rigid}^1$ is rigid in and of itself, the only way it can reorient is to invert. Now, assume the inversion happens over $p_1$. In this case, $p_2$ must now be flipped over the plane that $p_1$ exists in. Since $p_2$ is also part of $comp_{rigid}^2$ and is in a new location, $comp_{rigid}^2$ must also have reoriented to avoid breaking bonds with $p_2$, the only way to do which would be inverting. Now, in the case that the inversion happens over another face, the plane that contains that face cannot contain both $p_1$ and $p_2$ (since they are not coplanar by the claim) and therefore at least one of the shared faces ($p_1$ or $p_2$) will be in a new location. This means that $comp_{rigid}^2$ still had to reorient. Since both of these cases lead to contradictions, $comp_{rigid}^1$ and $comp_{rigid}^2$ must be part of the same rigid component.
\end{proof}

This provides us another powerful tool, showing that two shared non-coplanar faces is all that is required to entangle two rigid components into one.

\begin{figure}
\centering
\includegraphics[width=\textwidth]{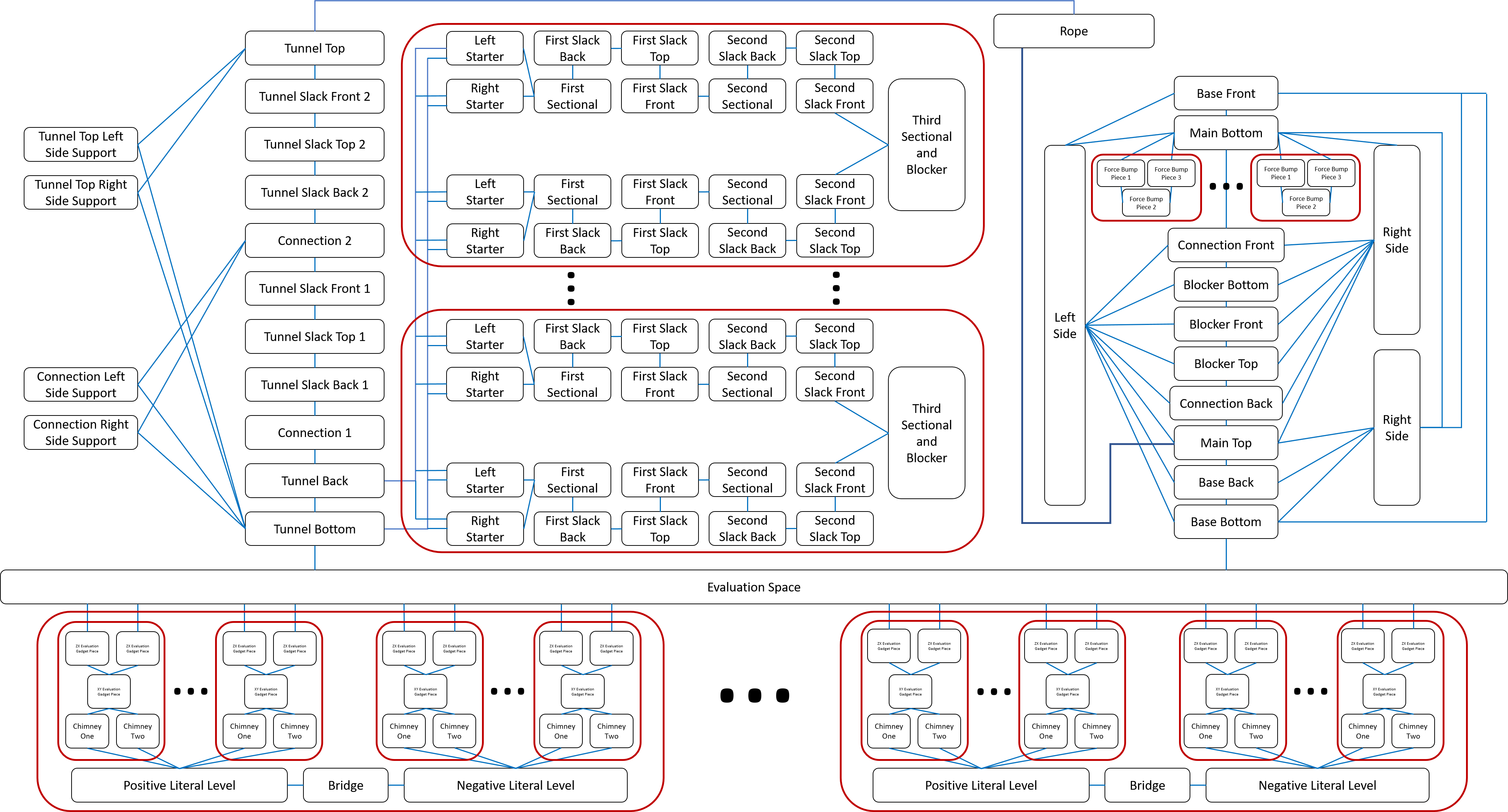}
    \caption{The full layout of the 3SAT machine. Each box highlights a section of the assembly that is duplicated with respect to the size of the corresponding 3SAT formula.}
    \label{fig:machine_layout}
\end{figure}

Now that we have some tools to use in our proof, we give a proof overview. Figure~\ref{fig:machine_layout} shows the machine assembly broken down into individual components in the form of a face graph. The highlighted pieces indicated the components that are duplicated as the corresponding 3SAT problem has more clauses and more variables. The first step in the proof will be to take some of these pieces and entangle them together into the major rigid components of our construction. These include the evaluation space, variable constraint gadgets, trival assignment hat, and satisfying assignment hat. Figure~\ref{fig:sat_machine_rigid_components} shows how these rigid components will be made, while Figure~\ref{fig:sat_machine_reduced} shows what the reduced layout will look like when considering these larger pieces. Finally, we will show that the pieces interact in such a way that the rigidity of the machine gives a valid reduction from the 3SAT problem.

\begin{figure}[htp]
\centering
    \subfloat[][The layout of the machine with all faces included. The faces that are in the same box will be entangled together into rigid components.]{%
        \label{fig:sat_machine_rigid_components}%
        \includegraphics[width=\textwidth]{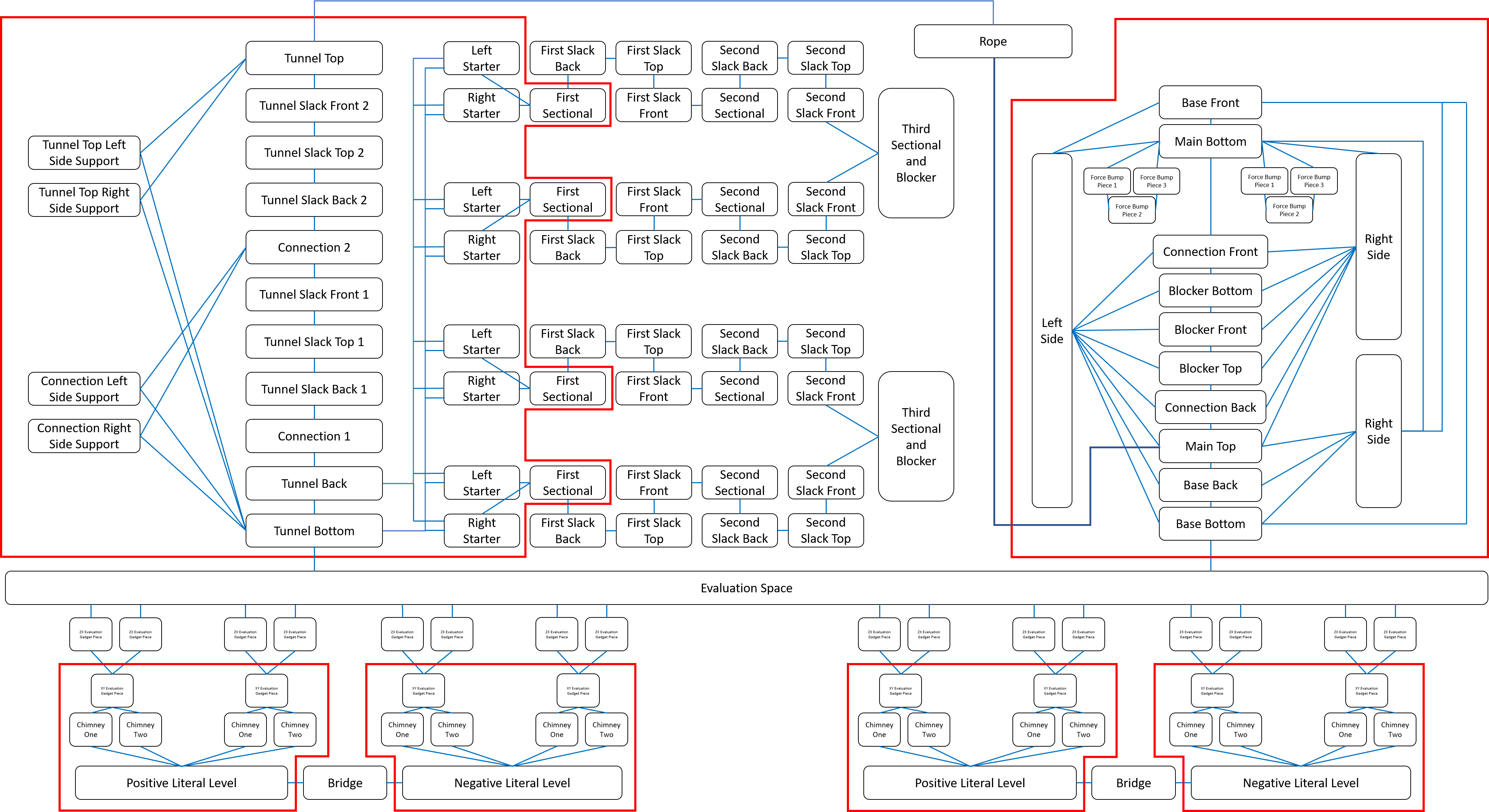}
        }
        \quad
    \subfloat[][The layout when considering the rigid components as one piece. The remaining flexible sections are the only ones that can reorient and only in the case that the corresponding 3SAT instance has a satisfying assignment.]{%
        \label{fig:sat_machine_reduced}%
        \includegraphics[width=0.75\textwidth]{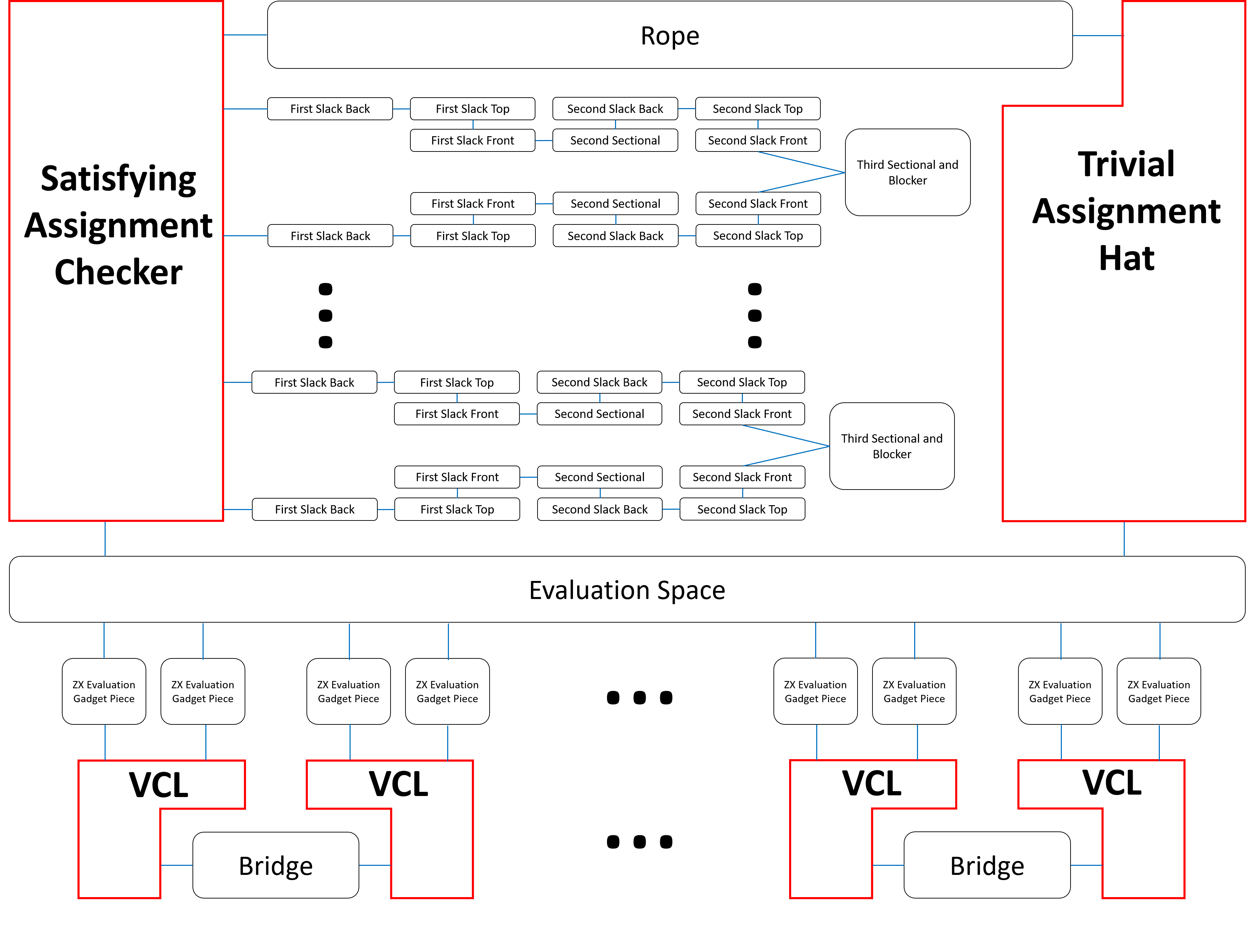}
        }
  \caption{The before and after diagrams with respect to the entanglement of rigid components.}
  \label{fig:machine_reduction}
\end{figure}

\subsubsection{Evaluation Space (ES)} \label{evaluationSpace}

The evaluation space, shown in Figure~\ref{fig:evaluation_space}, is a frame that facilitates interactions between the variable constraint gadgets, trivial assignment hat, and satisfying assignment hat. The frame consists of three \emph{evaluation gadgets} per row (which we colloquially refer to as ``bumps''), with as many rows as clauses in the corresponding 3SAT instance. The evaluation gadget consist of three tiles that form a bump over or under the open holes in the evaluation space. The three tiles that compose the evaluation gadget consist of two tiles in the $ZX$ plane and a tile in the $XY$ plane connecting them. An evaluation gadget forced up (Figure~\ref{fig:variable_false}) represents a literal evaluated to \emph{False}, while an evaluation gadget forced down (Figure~\ref{fig:variable_true}) represents a literal evaluated to \emph{True}. In the situation where the corresponding 3SAT instance does not have a satisfying assignment, the evaluation gadgets will collide with tiles from the satisfying assignment hat, preventing any configuration other than the initial. Every variable in the 3SAT instance has a corresponding evaluation gadget unless that variable has only positive literals or negative literals. If this is the case, we exclude that particular evaluation gadget since all instances of that variable can be assumed to be \emph{True}, so it is unnecessary for that evaluation gadget to ever collide with the satisfying assignment hat.

Attached to the evaluation space on opposite sides are the satisfying assignment hat and the trivial assignment hat. These pieces are attached to the evaluation space in such a way that exactly one of the them will be pressed against the evaluation space in any possible configuration. The satisfying assignment hat can only reorient down onto the evaluation space (as shown in Figures \ref{fig:satisfying_state}, \ref{fig:satisfied_in_simulator_side}, and \ref{fig:satisfied_in_simulator_top}) if and only if there is a solution to the corresponding 3SAT instance; otherwise, the configuration where the trivial assignment hat is pressed against the evaluation space (as shown in Figures \ref{fig:trivial_state} and \ref{fig:trivial_in_simulator}) will be the only valid configuration that the assembly can exist in.

\begin{figure}[htp]
\centering
    \subfloat[]{%
        \label{fig:variable_false}%
        \includegraphics[width=0.33\textwidth]{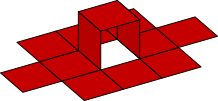}
        }%
    \subfloat[]{%
        \label{fig:variable_true}%
        \includegraphics[width=0.33\textwidth]{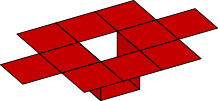}
        }%
    \subfloat[]{%
        \label{fig:formula}%
        \includegraphics[width=0.33\textwidth]{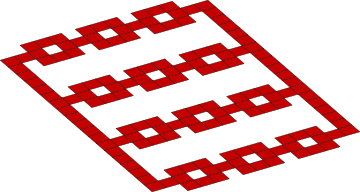}
        }
  \caption{The setup of the evaluation space. The tiles that make up the evaluation gadgets are included only in (a) and (b) to show the difference between a \emph{False} variable and a \emph{True} variable, respectively. Three individual evaluation gadgets together make up a clause, and multiple clauses together make up a formula, shown in (c).}
  \label{fig:evaluation_space}
\end{figure}

\subsubsection{Variable Constraint Gadgets (VCG's)} \label{variableConstraints}

The variable constraint gadgets sit below the evaluation space (in the $-Z$ direction) and interact with the evaluation gadgets to ensure that each instance of a variable, even the negated instances, agree with each other. Only variables that have at least one positive literal and at least one negative literal will have a variable constraint gadget associated with that specific variable. Since other variables with all positive or all negative instances have no evaluation gadgets, they also do not have variable constraint gadgets.

The variable constraint gadgets are designed such that all eligible variables are assigned a unique level at a different $Z$ value below the evaluation space. Each positive / negative instance of a single variable has two parallel strings of tiles called \emph{chimneys} that connect the $XY$ tile of all evaluation gadgets that correspond to that positive / negative variable to the \emph{variable constraint level} (VCL) of that positive / negative variable. For example, all evaluation gadgets for $x_1$ literals have chimneys that extend to the same variable constraint level at a specified $Z$ value, as well as all evaluation gadgets for $\overline{x_1}$ literals.

The variable constraint levels, shown in Figure~\ref{fig:variable_constraint_gadget}, connect all instances of a variable. There are two levels per variable, one for all positive literals, i.e. $x_1$, and one for all negative literals, i.e. $\overline{x_1}$. The two levels are then connected by a \emph{bridge}, which is attached to the end of both \emph{crossbars}. The bridge is simply a domino that connects the two levels in the $ZX$ plane. Its purpose is to ensure that the levels exist in $XY$ planes that are two units in the $Z$ direction apart. Since the levels are connected to the chimneys and to the respective evaluation gadgets, we will show that this ensures that the evaluation gadgets for positive and negative literals of the same variable always disagree.

\begin{figure}[htp]
\centering
    \subfloat{
        \label{fig:variable_contstraint_level}%
        \includegraphics[width=0.8\textwidth]{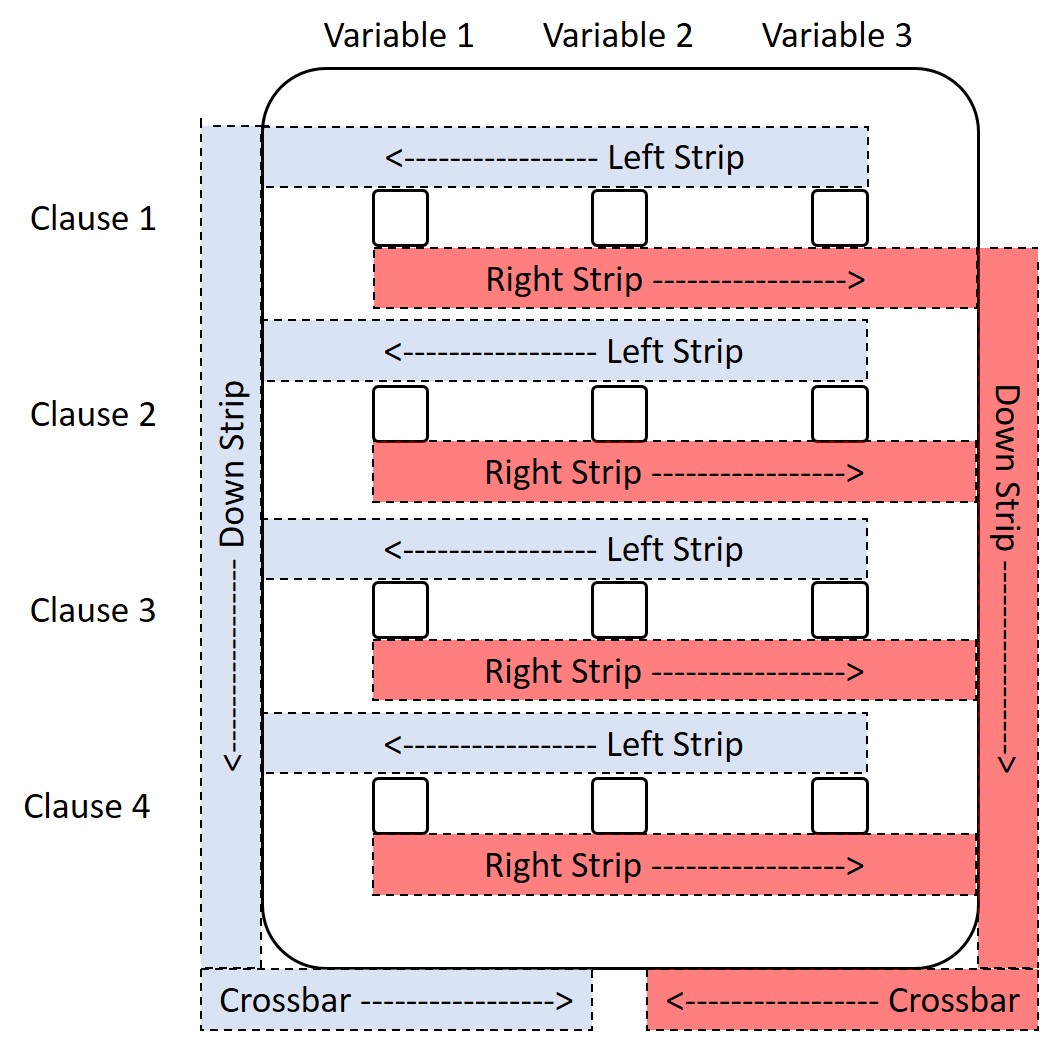}
        }
    \subfloat{
        \label{fig:vcl_compass}
        \includegraphics[width=0.2\textwidth]{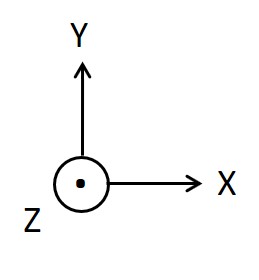}}
  \caption{The layout of a pair of variable constraint levels (VCL's). Each variable constraint gadget (VCG) has a ``positive'' level (blue) and a ``negative'' level (red). Each evaluation gadget is connected to its corresponding VCL by a vertical string of tiles (the chimneys) that descends in the $Z$ direction to the VCL designated for that variable. Once in that plane, the tiles within each VCL extend down the strips to the crossbars. The crossbars of two opposite levels are connected by a two tile long bridge.}
  \label{fig:variable_constraint_gadget}
\end{figure}

\begin{claim} \label{clm:vcl_rigid}
    The variable constraint gadgets can only reorient by changing truth value or being inverted.
\end{claim}

\begin{proof}
    To recap, for any variable $x$ with $p \ge 1$ instances of the positive literal and $n \ge 1$ instances of the negated literal, the variable constraint gadget is made up two levels (one for $x$ and one for $\overline{x}$), $p$ pairs of chimneys to one level, connected to the evaluation space collectively by $p$ evaluation gadgets, $n$ pairs chimneys to the other level, connected to the evaluation space collectively by $n$ evaluation gadgets, and one bridge between the two levels. Evaluation gadgets are made up of three pieces, one in the $XY$ plane and two in the $ZX$ plane.
    
    First, notice that an $XY$ piece of an evaluation gadget, a pair of chimneys, and the level are connected to make a 4 sided loop. Take anyone of these loops and assume it is rigid by Claim~\ref{clm:one_piece_rigid}. Obviously, the orientation of every other chimney loop that exists between the evaluation space and that specific level must also remain unchanged in order for them to connect to the level. Now, we have a rigid component of one level, all the chimneys attached to it, and $XY$ pieces of their respective evaluation gadgets. This rigid component can move up or down by two units due to reorientation of the $ZX$ pieces of the evaluation gadgets. We apply the same argument to level, chimneys, and evaluation gadgets corresponding to $\overline{x}$.
    
    Now, our variable constraint gadget for $x$ consists of two entangled rigid components (which we will now call $l$ and $r$), two sets of $ZX$ pieces from the evaluation gadgets (one connected to $l$ and one connected and $r$), and the bridge. We already know that a reorientation can happen that moves either $l$ or $r$ up two units in the $Z$ direction, the other down two units in the $Z$ direction, and causes a transformation of all $ZX$ pieces from the evaluation gadgets and the bridge. However, this is the only reorientation that can happen. To show this, notice that any pair of $ZX$ pieces for the evaluation gadgets must agree to be connected to the evaluation space through both $U$ bonds or both $D$ bonds in order to have a gap of one between them for the $XY$ piece to fit into. Since the rigid component $l$ or $r$ has already been shown to be rigid from one $ZX$ evaluation piece to another, the orientation of this pair of $ZX$ evaluation pieces will force the orientation of all other pair of evaluation pieces also connected to the $l$ or $r$ piece. The same argument goes with the $ZX$ evaluation pieces connected to the other $l$ or $r$ rigid component. Notice that the $l$ and $r$ component can't be on the same level, otherwise the two ends connected by the bridge would be incident on each other, leaving no room for the bridge itself. Therefore, one rigid component (either $l$ or $r$) must have each of its $ZX$ evaluation pieces be up, while the other rigid component must have each of its $ZX$ evaluation pieces be down. Therefore, the only ambiguity in the configuration of the variable constraint gadget (other than the trivial ambiguity of inversion) is determining which rigid component between $l$ and $r$ is up and which rigid component is down.
    
    Because all variable constraint gadgets have the same structure, this argument holds for each variable constraint gadget corresponding to each variables in the formula.
\end{proof}

\subsubsection{Trivial Assignment Hat (TAH)} \label{trivialAssignment}

 The purpose of the trivial assignment hat is to prevent the reorientation of any piece in the assembly as long as it is pressed against the evaluation space. It does this in two ways. First, it has what we labeled as ``Force Bumps'' that force the orientation of all evaluation gadgets and variable constraint gadgets, as shown in Figure~\ref{fig:tah_inkscape}. These are designed such that all negative literals in the evaluation space have a corresponding Force Bump on the trivial assignment hat and all positive literals do not. This effectively assigns a value of \emph{False} to all variables in the formula, hence the name ``trivial assignment hat'', forcing the evaluation gadgets and variable constraint levels of all negative variables down and thereby forcing the evaluation gadgets and variable constraint levels of all positive variables up.
 
 The second way the trivial assignment hat disables reorientation is by having an additional piece that blocks otherwise-free pieces in the satisfying assignment hat from moving. This piece can be seen in Figure~\ref{fig:tah_simulator}. It will later be proven that this is sufficient to show that the machine is rigid in the case when no satisfying assignment exists to the corresponding 3SAT problem.

\begin{figure}[htp]
    \subfloat{%
        \label{fig:tah_inkscape}
        \includegraphics[width=0.5\textwidth]{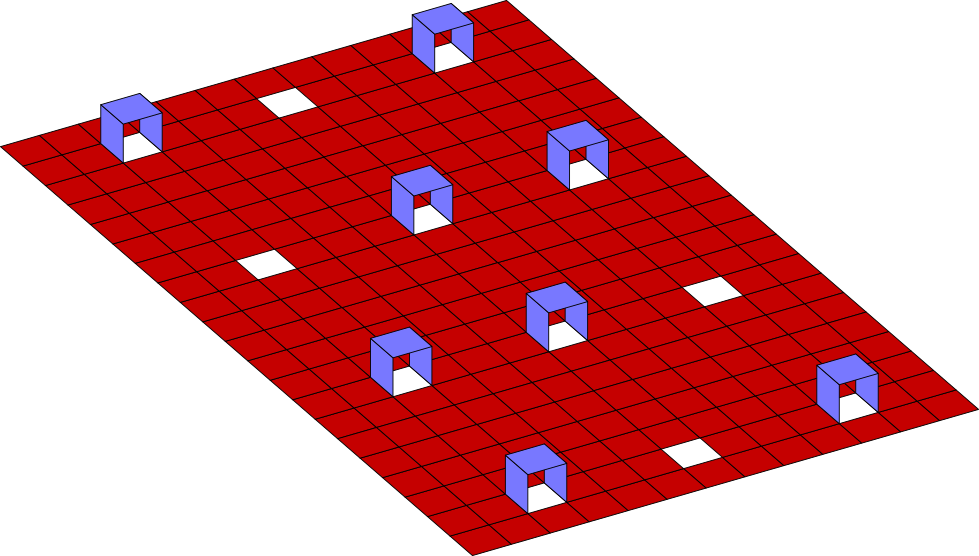}
        }%
    \subfloat{%
        \label{fig:tah_simulator}
        \includegraphics[width=0.5\textwidth]{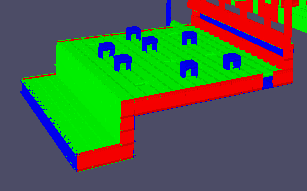}
        }
    \caption{The layout of the trivial assignment hat. The left image just shows the Main Bottom piece along with the Force Bumps. The right is an image take from FTAM visualizer. The bumps are aligned with variables in the evaluation space. From here, the gadget would be flipped over and pressed onto the top of the evaluation space to force the orientation of any variables that were free to move. The additional tiles at the bottom left of the right image are used to block the satisfying assignment checkers whenever the assembly is in the trivial state.}
    \label{fig:trivial_assignment_hat}
\end{figure}

\begin{claim}
    The trivial assignment hat can't reorient without being inverted.
\end{claim}

\begin{proof}
    To clarify some of the terminology of this section refer to Figure~\ref{fig:tah_labeled}, Figure~\ref{fig:machine_layout}, or Figure~\ref{fig:machine_reduction}.

    Here we will show that most faces in the trivial assignment hat are part of a traditional four sided loop utilizing the ``Left Side'' piece and one of the two ``Right Side'' pieces, i.e. every loop will look like: Left Side, piece one, Right Side, piece two, Left Side. The planar sections that will be substituted into ``piece one'' and ``piece two'' are as follows: (Base Bottom, Main Top), (Base Front, Base Back), (Main Bottom, Main Top), (Main Bottom, Blocker Top), (Connection Front, Connection Back), (Blocker Front, Blocker Back), (Blocker Bottom, Blocker Top). By Claim~\ref{clm:two_piece_entanglement}, since each loop contains Left Side and one of the two Right Side's, we are left with two rigid components (corresponding to the split up of the Right Side's). However, again by Claim~\ref{clm:two_piece_entanglement}, both rigid components share Left Side, Main Top, and Main Bottom, none of which are coplanar, and can therefore be entangled into one rigid component. The only remaining planar sections in the trivial assignment hat are the Force Bumps, each of which forms a traditional four sided loop with the Main Bottom. Since each of these loops is rigid, the piece that is opposite of the Main Bottom must be in the same orientation as Main Bottom but in a plane that is one unit away in either normal direction. Since the Main Top is in the same orientation and in the plane that is one unit above the Main Bottom, no Force Bump can invert because it would collide with the Main Top. Therefore, each Force Bump can be entangled into the rigid component, leaving us with one rigid trivial assignment hat.
\end{proof}

\begin{figure}[htp]
\centering
    \subfloat{%
        \label{fig:tah_labeled}
        \includegraphics[width=\textwidth]{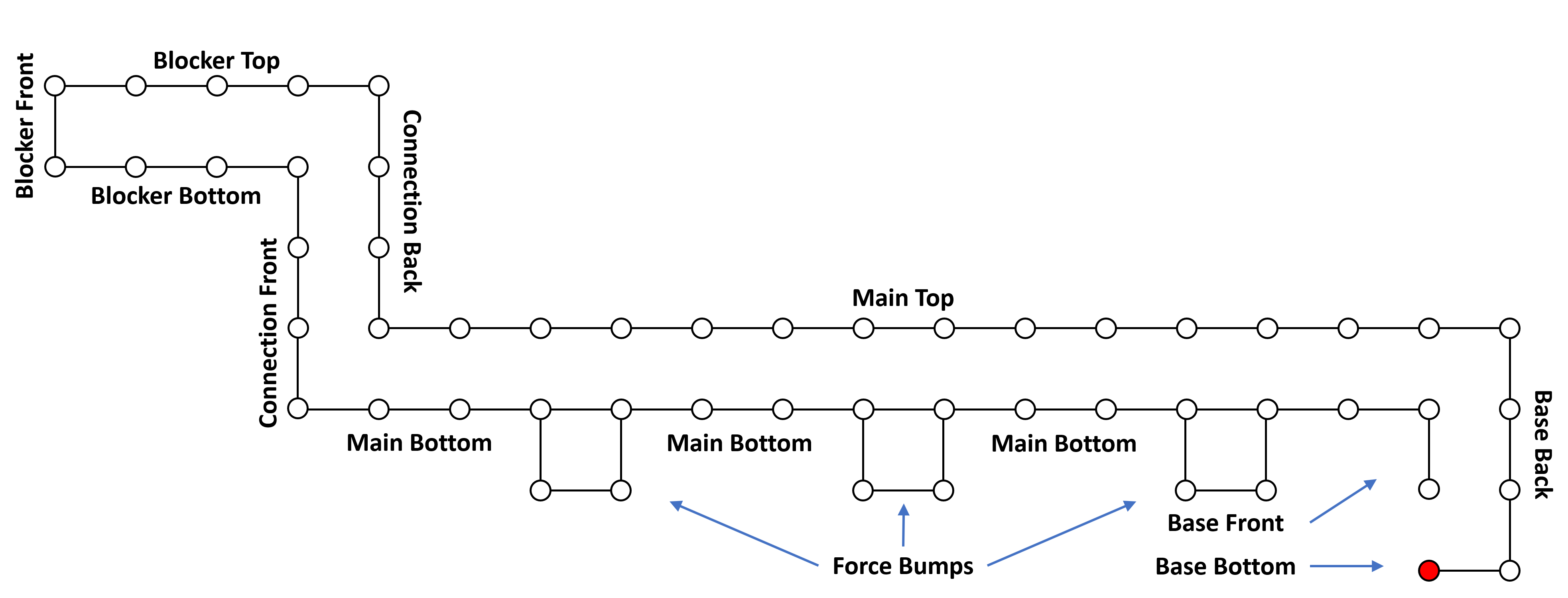}
        }
    \quad
    \subfloat{%
        \label{fig:sah_labeled}
        \includegraphics[width=\textwidth]{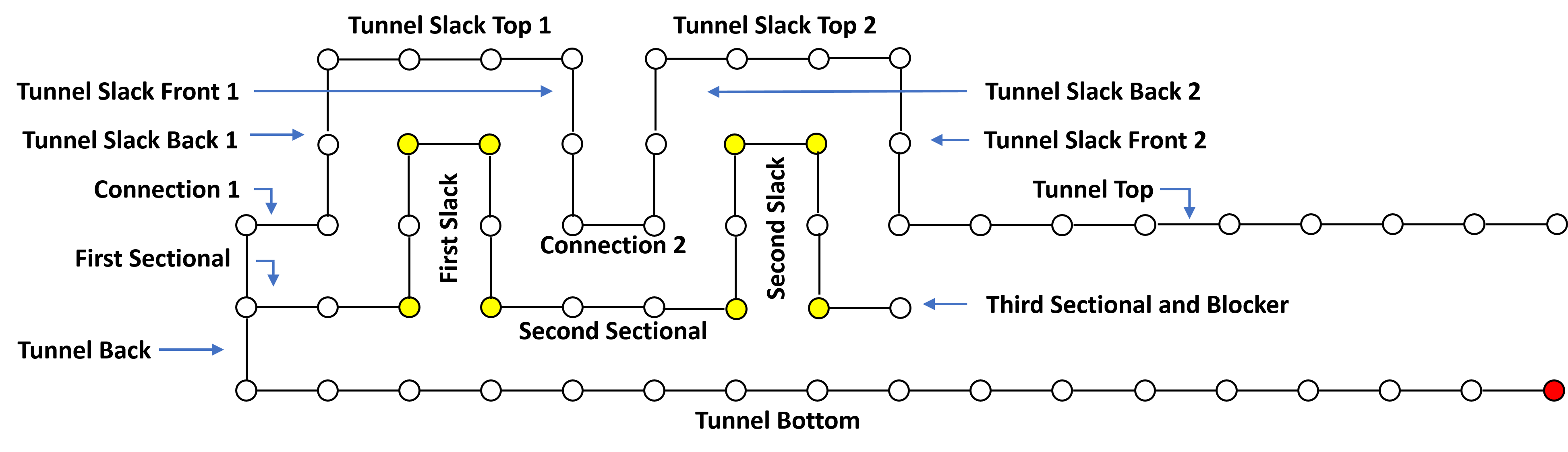}
        }
    \quad
    \caption{Labeled images showing the names of each face in the two assignment hats.}
    \label{fig:hat_labels}
\end{figure}

\subsubsection{Satisfying Assignment Hat (SAH)} \label{satisfyingAssignment}

The satisfying assignment hat comes down and presses up against the evaluation space to make an additional valid configuration for the whole machine if and only if there is a solution to the inputted 3SAT problem. The satisfying assignment hat works by using a \emph{checker} for each clause in the formula. The checker is encased in a structure, shown in Figure~\ref{fig:SAH_side_view}, that allows it to be in three different orientations, shown in Figure~\ref{fig:SAH_top_view}. These different orientations each project a blocking tile, shown in yellow, to a different location to block one of the three variables in the clause. The checkers can only be in three different orientations due to the structure the checker is encased in. In Figure~\ref{fig:SAH_side_view}, we can see this illustrated by the rigid outer structure and the partly flexible checker inside the structure. Marked in yellow, the flexible points in the checker allow for the checker to only ``fold into'' one of the available, raised parts of the structure. If the checker were to fold into a non-raised portion of the structure, the tiles would simply conflict with the outer structure, shown in red in  Figure~\ref{fig:SAH_side_view}. If the satisfying assignment hat is able to press against the evaluation space, that means the checker for each clause has found at least one of the variables to be \emph{True} in each clause, allowing the checker to be in the orientation that projects the blocking tile over the evaluation gadget that corresponds to that variable.

\begin{figure}[htp]
\centering
    \includegraphics[width=\textwidth]{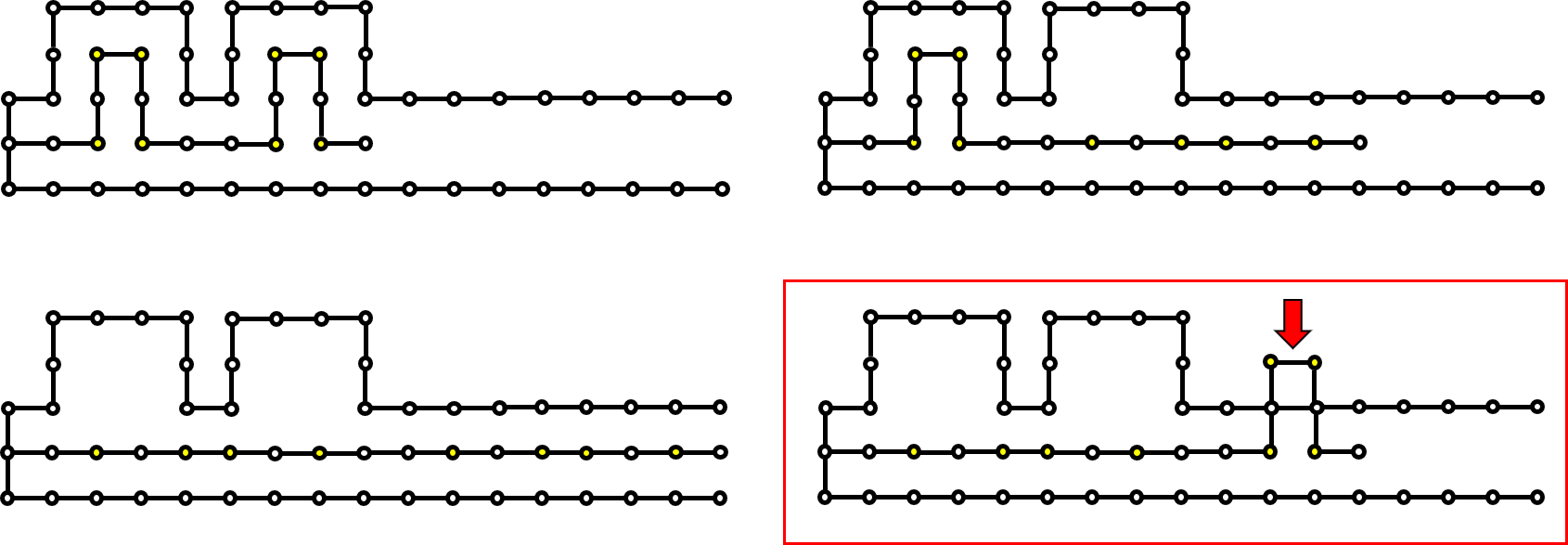}
    \caption{The structure that allows the checker piece to move to one of three different locations. The bottom right orientation is invalid because the checker piece conflicts with the outer structure.}
    \label{fig:SAH_side_view}
\end{figure}

\begin{figure}[htp]
\centering
    \includegraphics[width=0.5\textwidth]{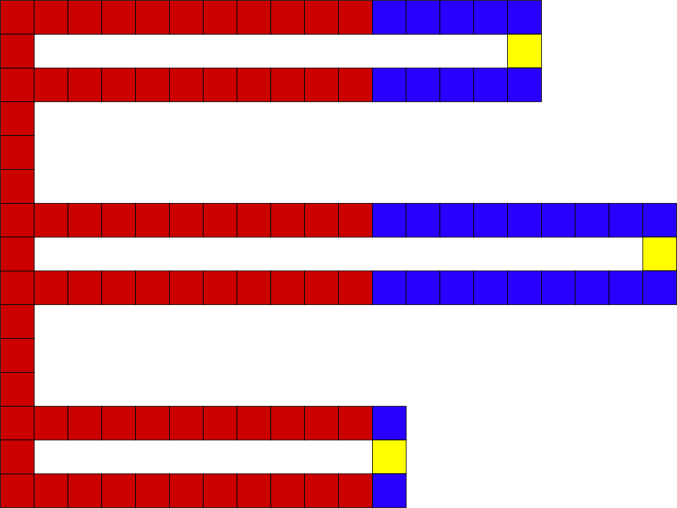}
    \caption{A top view of a satisfying assignment hat with the three checkers occupying the three different possible positions.}
    \label{fig:SAH_top_view}
\end{figure}

\begin{claim}
    The outer structure of the satisfying assignment hat can't reorient without being inverted.
\end{claim}

\begin{proof}
    To clarify some of the terminology of this section refer to Figure~\ref{fig:sah_labeled}, Figure~\ref{fig:machine_layout}, or Figure~\ref{fig:machine_reduction}.

    First, notice the traditional 4-sided loop between the Tunnel Bottom and Tunnel Top utilizing the Tunnel Top Side Support's. Similarly, the Tunnel Bottom and Connection 2 piece form a 4-sided loop. We can entangle these two loops because, if one were to invert without the other, the three Tunnel Slack~2 pieces would have to collide with the Tunnel Bottom piece in order to connect the two. Since the Connection 2 piece and the Tunnel Top are part of the same rigid component, we can also entangle the three Tunnel Slack 2 pieces. Now, notice that the Tunnel Bottom, Tunnel Back, and any Left Starter or Right Starter piece form a vertex. This means that, in relation to the Tunnel Bottom, the Tunnel Back must either be ``Up'' or ``Down'' but cannot be ``Straight''. However, if it were to be ``Down'' (thereby inverted relative to the rigid component), its connection to the Tunnel Back would be $4$ down in the $Z$ dimension and $4$ to the left in the $X$ dimension to the connection between the Tunnel Slack Front 1 and the Connection 2. Since this Manhattan distance is $8$ and the combined length of the four piece that need to make up this distance (Connection 1, Tunnel Slack Back 1, Tunnel Slack Top 1, and Tunnel Slack Front~1) is also $8$, the pieces would have to form a shortest path from one connection to the other. However, since every shortest path would intersect with either the Tunnel Bottom or Tunnel Back, it cannot make this distance. Therefore, the Tunnel Back (and Left Starters, Right Starters, and First Sectionals) can be entangled into the rigid component. Now, the Manhattan distance that needs to be made up by the 4 pieces is just $4$. Since the respective lengths of the 4 pieces in the order previously given are $1$, $2$, $3$, $2$, the only arrangement that would give a total displacement of $4$ would for the $2$ length pieces (Tunnel Slack Back 1 and Tunnel Slack Front 1) to cancel out and for the other two pieces (Connection~1 and Tunnel Slack Top 1) to make up the distance. This leaves us with two options, the back and front slack folded away from the Tunnel Bottom or towards the Tunnel Bottom. As mentioned earlier, if the slacks fold towards the Tunnel Bottom, the Tunnel Slack Top 1 would collide with the Tunnel Bottom. Therefore, there is only one configuration of these pieces with respect to the rigid component, meaning they can be entangled, leaving us with one big rigid component.
\end{proof}

\subsubsection{3SAT Checker Machine} \label{3SAT}

There is one last piece of the assembly we must introduce, which we refer to as the \emph{rope}. This is a $15 \times 1$ face of tiles that connects the trivial and satisfying assignment hats. It connects to the Tunnel Top on the SAH and the Main Top on the TAH. It's purpose is to prevent the invalid configuration from Figure~\ref{fig:invalid_state} where neither hat is pressed against the evaluation space.

Now that we have introduced all the pieces of the assembly, we need to focus on how the remaining flexible bonds can move relative to one another. We group the remaining flexible bonds into three groups: the main loop bonds, the checker bonds, and the variable constraint bonds. The main loop group consists of four bonds: the ES to the SAH, the SAH to the rope, the rope to the TAH, and the TAH to the ES. Notice that these bonds form a loop. By looking at Figure~\ref{fig:main_loop}, you will notice this loop can only take two configurations, ``UDSS'' and ``SSDU'' (with respect to ES-SAH, SAH-rope, rope-TAH, TAH-ES, with the normal of the rope pointed in the $+Z$ dimension and the normal of the ES pointed in the $-Z$ dimension). We refer to these two configurations as the trivial state and satisfied state, respectively.

\begin{figure}[htp]
\centering
    \includegraphics[width=\textwidth]{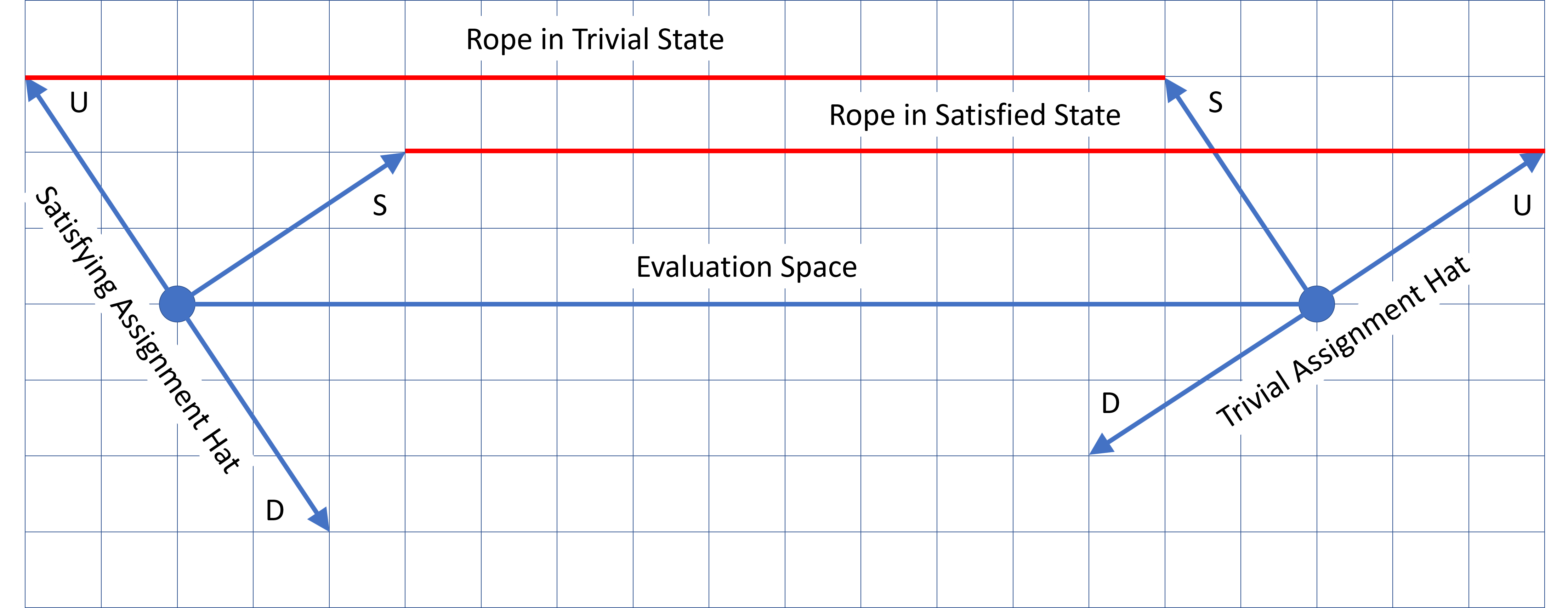}
    \caption{A view of the main loop of the assembly from the negative $Y$ direction. The arrows represent the displacement between the bonds from the evaluation space to each hat and the bonds from each hat to the rope. The orientation labels are for the bonds between the evaluation space and each hat. The illustration is meant to show that the only positions that the two hats can be in such that their bonds to the rope are in the same $Y$ plane and exactly 15 units apart in the $X$ plane are the two labeled positions for the rope.}
    \label{fig:main_loop}
\end{figure}

The checker group consists of 16 bonds for each checker gadget in the SAH. These 16 bonds can be viewed as two subgroups of 8 which must be the same configuration from the set of ``UDDUUDDU'', ``UDDUSSSS'', and ``SSSSSSSS'' (with respect to the sequences of bonds starting at the First Sectional piece and ending with the Third Sectional and Blocker piece, with the normal of these two pieces pointed at the Tunnel Top). The sequences for the two subgroups of any checker gadget must match since both sequences of bonds end with the Third Sectional and Blocker piece, which connects with both previous pieces at the same $X$ coordinate. As opposed to a formal proof, we direct the reader to Figure~\ref{fig:SAH_side_view} to see the reasoning that these are the only three possible configurations. Note that the invalid configuration shown in the illustration would be ``SSSSUDDU''.

The variable constraint group consists of all the bonds between the ES and the $ZX$ evaluation gadget pieces, all the bonds between the $ZX$ evaluation gadget pieces and the $XY$ evaluation gadget pieces, and all the bonds between the VCL's and the bridges. As we saw in Claim~\ref{clm:vcl_rigid}, this entire group of bonds can only have as many possible configurations as there are possible assignments to the variables in the corresponding boolean formula, i.e. $2^n$ where $n$ is the number of variable in the formula. An intuitive explanation of this is that each variable has a variable constraint gadget, each of which can only take two states, one where the variable is \emph{False} and one where the variable is \emph{True}.

Finally, we show what can happen when the main loop is in the trivial state and satisfied state.

\begin{claim}
    Neither the trivial nor the satisfying assignment hat can invert without the other.
\end{claim}

\begin{proof}
    Notice that the displacement vectors in Figure~\ref{fig:main_loop} between the bond where either hat connects to the ES and the bond where the same hat connects to the rope is some rotation of the vector $(2,3)$ in the $ZX$ plane (it is the reverse of the apparent vector in the Figure, since the illustration is from the $-Y$ perspective instead of the $+Y$ perspective). If one hat were to invert without the other, the displacement vector of the inverted hat would change to some rotation of the vector $(3,2)$. Then, for the two displacement vectors to have the same $Z$ value (to have the ends of the rope be in the same $Z$ plane) the $X$ values must be the negation of each other. This leaves the possible distance between the two bonds at the ends of the rope to be $9$, $11$, $19$, or $21$. These numbers were computed by taking the length of the evaluation space, $15$, and adding or subtracting both of the possible $X$ coordinates of the displacement vectors, $2$ or $3$, since the displacement vectors must point in opposite $X$ directions. Since none of these possible distances are $15$, the actual distance of the rope, then none of the possible configurations resulting from an inversion of one hat without the other are valid.
\end{proof}

\begin{claim} \label{clm:trivial_rigid}
    If the main loop is in the trivial state, no other bond can reconfigure.
\end{claim}

\begin{proof}
    First, we also prove that the variable constraint gadgets can't invert without the SAH and TAH. This is easy to see from Figures~\ref{fig:simulater_machine} and \ref{fig:machine}. Since the TAH lays across the entire ES, if the variable constrain gadgets were to invert without the other main components, the chimneys would collide with the Main Bottom and Main Top pieces in the TAH. Therefore, the variable constraint gadgets can't reorient without both of the other main components reorienting as well.
    
    Now, we show the remaining flexible bonds (the checker group and the variable constraint group) can't reorient either. In the trivial state, the checkers in the SAH are blocked by the extension of the TAH, which prevents the bonds from taking any configuration other than ``UDDUUDDU''. The bonds in the variable constraint group are also forced into a configuration by the TAH, since the Force Bumps on the TAH force all evaluation gadgets for positive literals to be popped down, thereby forcing their respective VCL's down, thereby forcing their opposing VCL's up, and thereby forcing all evaluation gadgets for negative literals to be popped up. Therefore, if the main loop bonds are fixed in the trivial state, the entire assembly is rigid.
\end{proof}

\begin{claim} \label{clm:satisfied_flexible}
    If the main loop is in the satisfied state, the other free flexible bonds can only configure in a way such that no tiles in the assembly are overlapping if and only if there is a satisfying assignment to the corresponding 3SAT formula.
\end{claim}

\begin{proof}
    First, we address the inversion of variable constraint gadgets. Since the TAH is no longer pressed against the ES, the variable constraint gadgets could invert across the $XY$ evaluation gadget pieces to which that specific variable constraint gadget is bound. However, this inversion would not affect the functionality of the variable constraint gadget, since it would still cause all positive literals of one variable to be one truth value and all negative literals of that variable to be the other truth value. Therefore, we can ignore any potential inversion of these gadgets in this proof. For elegance, the Tunnel Top piece of the SAH could be extended such that it blocks these inversions.

    For the forward direction of the bijection, we prove that, ``If the assembly configures such that no tiles overlap, there is a satisfying assignment''. We know that, in the satisfied state, the Third Sectional and Blocker piece in each checker gadget of the SAH must occupy the tile location that is one unit in the $Z$ dimension above one of the three evaluation gadgets in the clause to which the checker gadget corresponds. Therefore, if no tiles overlap, one of the literals in that clause must have evaluated to \emph{True} so that it is popped down in the $Z$ dimension, allowing the checker gadget to occupy the space above. If this evaluation gadget represents a positive literal, we can start building the satisfying assignment by setting that variable to \emph{True}. If the evaluation gadget represents a negative literal, we set the variable to \emph{False}. We do this for each checker gadget. We know that we will never have to reassign a variable a new value, since the variable constraint gadgets ensure that every literal of the same type agrees with each other. After we go through each checker gadget, we have a variable assignment that satisfies every clause in the boolean formula. Any variable that hasn't been assigned can be given either truth value.
    
    For the reverse direction, we prove that, ``If there is a satisfying assignment, the assembly can configure such that no tiles overlap''. Here, we can take every variable in the satisfying assignment, and if it is assigned the value \emph{True}, we raise the VCL that corresponds to the negative literals of that variable (and pop up the connected evaluation gadgets) and lower the VCL that corresponds to the positive literals of that variable (and pop down the connected evaluation gadgets). If the variable is assigned the value \emph{False}, we do the opposite. Now, because we know this is a satisfying assignment, at least one evaluation gadget in each clause must be popped down, allowing each checker gadget to have at least one configuration that doesn't overlap with an evaluation gadget, giving us a valid configuration for the whole assembly.
\end{proof}

To prove Lemma~\ref{lem:rigidity-from-assembly-NP-hard}, we use the previous claims as follows. Take the 3SAT machine in the trivial state. By Claim~\ref{clm:satisfied_flexible}, if the corresponding 3SAT formula has a satisfying assignment, then the machine has at least one additional configuration in the satisfied state it can reconfigure into. If the corresponding 3SAT formula does not have a satisfying assignment, there are no valid configurations in the satisfied state. By Claim~\ref{clm:trivial_rigid}, there is only one configuration in the trivial state. Therefore, determining if the assembly has multiple valid configurations, the complement of rigidity-from-assembly, is polynomial time reducible from 3SAT and therefore NP-hard.

\subsection{Terminality-from-assembly is co-NP-complete: Technical details} \label{sec:assembly-terminality-append}

Here we provide technical details for the proof of Lemma~\ref{lem:terminality-from-assembly-NP-hard}.

\begin{proof}
We will again reduce this problem to 3SAT. Take the 3SAT machine used to prove Lemma \ref{lem:rigidity-from-assembly-NP-hard}. This assembly has a bond between a piece called the ``Rope'' and the TAH that is in a ``Straight'' orientation when the machine is in trivial state and in a ``Down'' orientation in the satisfied state. For both tiles that make up this bond, add a unique flexible glue that is on the side 90 degrees clockwise from the original bond on one tile and 90 degrees counterclockwise from the original bond on the other tile. This way, both unique flexible glues are pointed in the same dimension. In the trivial state, these glues are adjacent to two different tile locations that are adjacent themselves. However, if the assembly can reconfigure into the satisfied state, these glues become adjacent to a mutual tile location. Therefore, by adding a tile type to the system with the complements of both glues on adjacent sides such that it could bind in this location, we create a situation in which the assembly is not terminal if and only if the corresponding 3SAT problem has a satisfying solution. Therefore, the complement of terminality-from-assembly is NP-hard.
\end{proof}

\end{document}